\documentclass[12pt,titlepage]{article}
\usepackage{rsipacks} 
\usepackage{algpseudocode}
\usepackage{algorithm}

\newtheorem{theorem}{Theorem}[section] 

\theoremstyle{definition}

\theoremstyle{remark}

\sectionfont{\large \centering}
\subsectionfont{\normalsize \centering}
\crefname{claim}{claim}{claims}
\Crefname{claim}{Claim}{Claims}
\crefname{app-corollary}{corollary}{corollaries}
\Crefname{app-corollary}{Corollary}{Corollaries}
\crefname{app-definition}{definition}{definitions}
\Crefname{app-definition}{Definition}{Definitions}
\crefname{figure}{figure}{figures}
\Crefname{figure}{Figure}{Figures}
\crefname{lemma}{lemma}{lemmata}
\Crefname{lemma}{Lemma}{Lemmata}
\crefname{app-lemma}{lemma}{lemmata}
\Crefname{app-lemma}{Lemma}{Lemmata}
\crefname{app-proposition}{proposition}{proposition}
\Crefname{app-proposition}{Proposition}{Proposition}
\crefname{app-theorem}{theorem}{theorems}
\Crefname{app-theorem}{Theorem}{Theorems}

\newcommand{\ptl}{\partial}
\newcommand{\til}[1]{\widetilde{#1}}

\newcommand{\Z}{\mathbb{Z}}

\newcommand{\R}{\mathbb{R}}

\newcommand{\M}[1]{\mathbf{#1}}

\makeatletter
\@addtoreset{subfigure}{figure}
\makeatother

\begin{document}



\title{From Walking to Tunneling: An Investigation of Generalized Pilot-Wave Dynamics}

\author{
Akilan Sankaran\footnote{This work is a preprint.
} \\
Albuquerque Academy
\vspace{0.5in}\\
Under the direction of\\
\\
Diego Israel Chavez\\
Applied Mathematics Lab\\
Department of Mathematics,\\
Massachusetts Institute of Technology\\
\vspace{0in}
}


\date{
\today
}

\maketitle

\begin{abstract}
\begin{singlespace}

We investigate the ability of millimetric walking droplets to tunnel between cavities. By synthesizing experimental and theoretical analysis, we provide a framework for droplet tunneling mechanics in three spatial dimensions. We define a generalized Dirichlet-to-Neumann operator that allows us to explicitly characterize droplet and wave-field dynamics under highly intricate variable-topography systems, allowing for numerical simulations of droplet tunneling probabilities and macroscopic dynamical evolution to a greater degree of accuracy than existing models. Moreover, we demonstrate experimental droplet tunneling in complex cavity geometries and discuss many-droplet coupling in the context of tunneling observations.
\vspace{1in}
\end{singlespace}
\end{abstract}

\begin{spacing}{1.45}
\section{Overview of Pilot-Wave Hydrodynamics}

We will take a short walk through the recent dynamics of research on walking droplets. The observation of fluid bath destabilization under external oscillatory forcing by Faraday \cite{faraday1831xvii} catalyzed various experimental and analytical investigations of hydrodynamic systems that lie close to the \textit{Faraday instability threshold} \cite{oza2013trajectory, benjamin1954stability, douady1990experimental}. Fluid-mechanical systems slightly below the Faraday instability exhibit various dynamical phenomena that defy physical intuition \cite{douady1990experimental, kaydureybush, couder2005bouncing,  faria2017model}. In particular, Walker demonstrated in 1978 that a millimetric droplet of silicone oil can bounce indefinitely on a bath of the same fluid under the imposition of forcing close to the Faraday threshold \cite{walker1978drops}, challenging the expectation that surface tension effects would lead to coalescence of the two fluid bodies \cite{couder2005bouncing}. Such a \textit{bouncing state} occurs since the lubrication effect of dissipation beneath the droplet creates an air film with sufficiently strong lift forces to propel the droplet upwards, allowing for recurrent bouncing \cite{couder2005bouncing,  walker1978drops,fort2010path}. Each droplet impact on the air film establishes a crater at the fluid surface, which propagates to form a wave-field (Figure \ref{fig-wave-field-for-intro}) \cite{couder2005bouncing}; the robust dynamical properties of the coupling between a droplet and its associated wave-field form the crux of pilot-wave hydrodynamics \cite{bush2015pilot}.

As the Faraday instability threshold is approached from below, the bouncing state destabilizes into complex dynamical behavior (Figure \ref{fig-regime-diagram}) \cite{kaydureybush, couder2006single, gilet2009chaotic}. Although many of the emergent dynamical states are chaotic \cite{gilet2009chaotic, pucci2018walking}, regimes of stability exist; in particular, medium-sized millimetric droplets --- that is, droplets with diameter between 0.4 mm and 1 mm --- undergo a \textit{period-doubling} bifurcation, in which successive bounces alternate between large and small amplitudes \cite{couder2006single}. After further increasing the acceleration of external forcing, a subcritical drift bifurcation destabilizes a bouncing droplet in the period-doubled state into horizontal motion, termed a \textit{walking state} \cite{couder2006single, rahman2020walking}. To establish this state, a slight perturbation in the droplet position leads to it landing on the sloping portion of its associated wave \cite{oza2013trajectory}, creating a propulsive force that drives the walking droplet forward in a parabolic vertical trajectory \cite{molavcek2013drops}.

The walking droplet model serves as an exemplar of a fluid-dynamical setup exhibiting sufficient stability to be analytically tractable \cite{faria2017model, molavcek2013dropsb} while also lying sufficiently close to the instability threshold to exhibit counterintuitive characteristics \cite{douady1990experimental, kaydureybush}. We center our analysis on three specific characteristics of walking droplet dynamics. Firstly, the wave-field created by a walking droplet effectively stores the previous positions of the droplet within its topography \cite{bush2015pilot, eddi2011information}; therefore, a walking droplet is closely coupled to its associated, delocalized waves through \textit{path-memory} \cite{oza2013trajectory, bush2015pilot, eddi2011information}. Second, walking droplets exhibit analogies to quantum-mechanical behaviors: droplets have exhibited diffraction properties through single- and double-slit structures embedded into the fluid bath surface \cite{couder2006single, pucci2018walking}, quantum tunneling analogues by passing through cavity boundaries \cite{nachbin2017tunneling}, and separation into orbital states \cite{eddi2012level}. The development of these surprisingly acute correspondences could aid in better comprehending quantum-mechanical behaviors from a macroscopic standpoint \cite{bush2015pilot}. Finally, the simplicity of the pilot-wave dynamics setup offers the possibility of theoretical and numerical descriptions of droplet behaviors in addition to experiment \cite{bush2015pilot, couder2006single, molavcek2013drops, molavcek2013dropsb}.

\begin{figure}[!t]
    \centering
\begin{subfigure}[t]{0.5\textwidth}
\centering
  \includegraphics[height=1.4in]{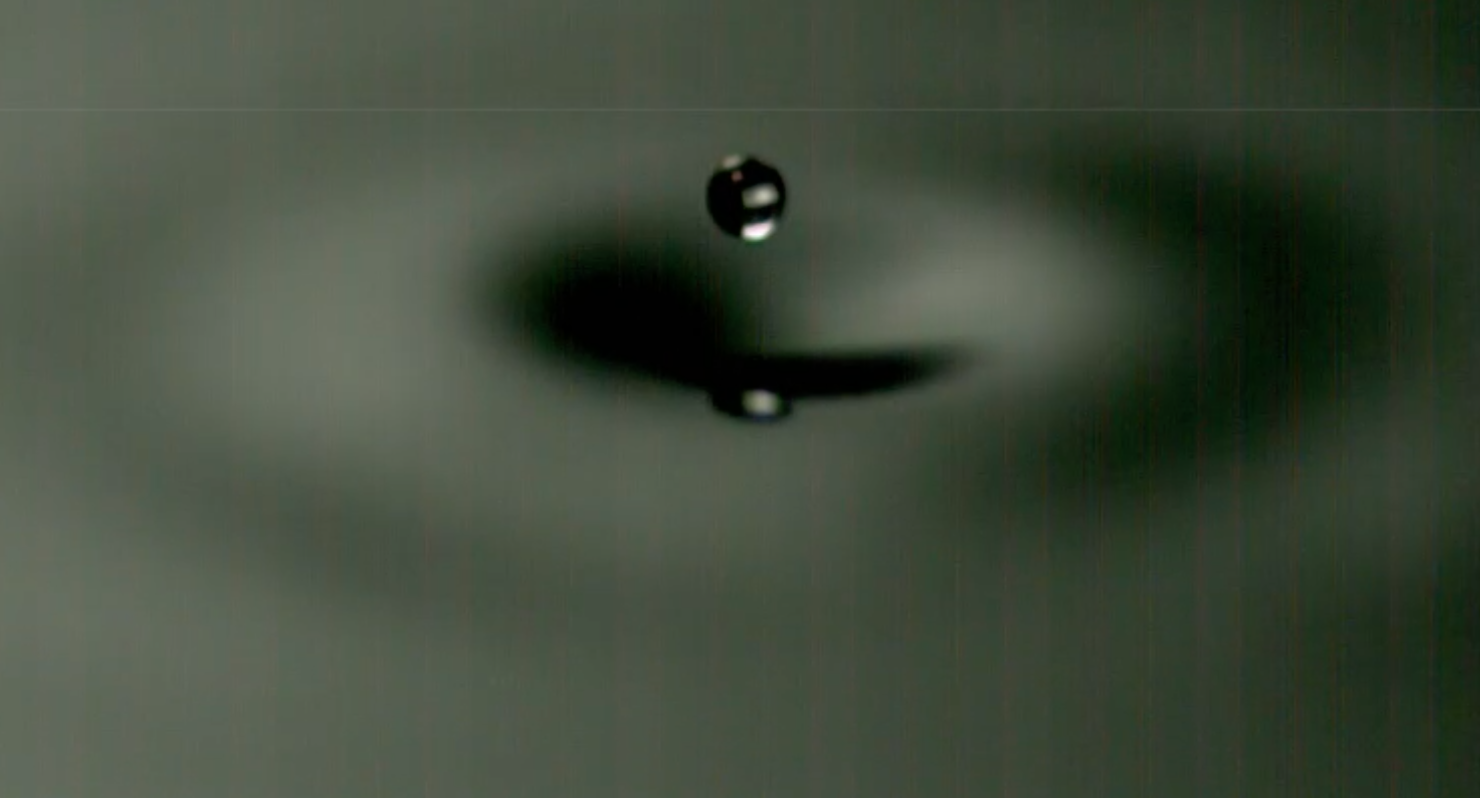}
\caption{Droplet position at maximum height.}
\end{subfigure}\hfill
\begin{subfigure}[t]{0.5\textwidth}
\centering
    \includegraphics[height=1.4in]{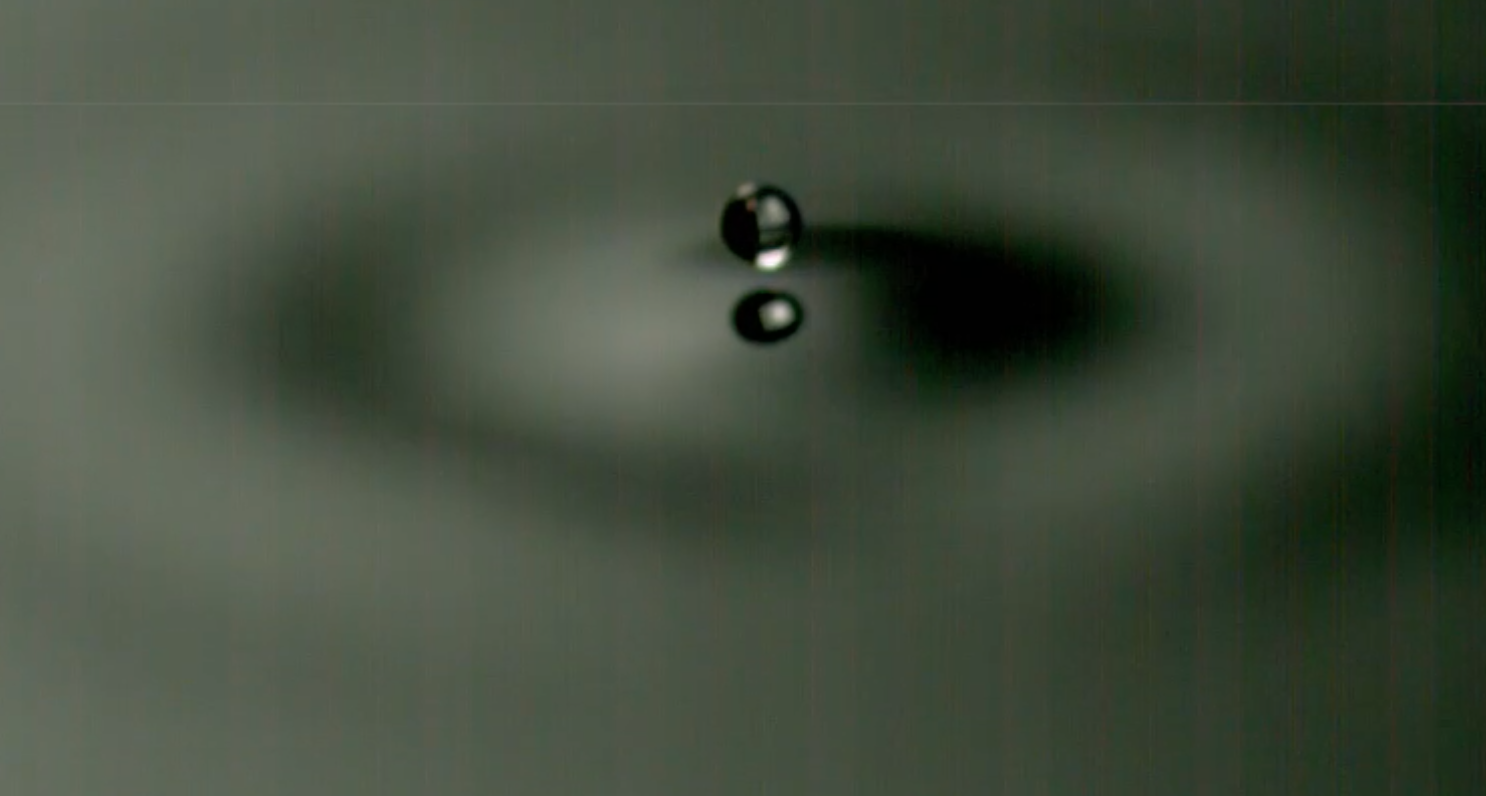}
\caption{Droplet descent --- lubrication effect.}
\end{subfigure}

\begin{subfigure}[t]{0.5\textwidth}
\centering
    \includegraphics[height=1.4in]{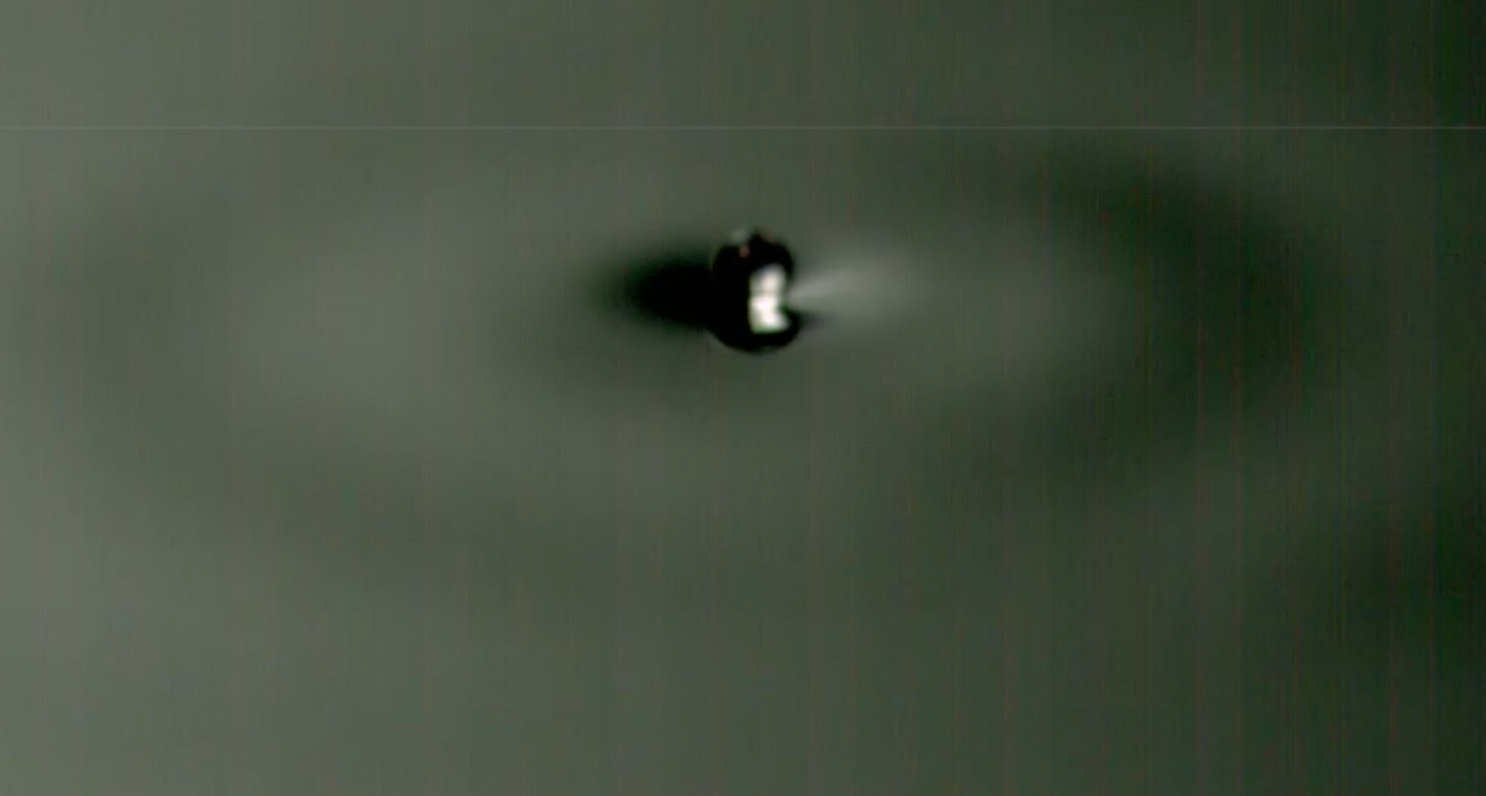}
\caption{Droplet-wavefield interaction.}
\end{subfigure}\hfill
\begin{subfigure}[t]{0.5\textwidth}
\centering
    \includegraphics[height=1.4in]{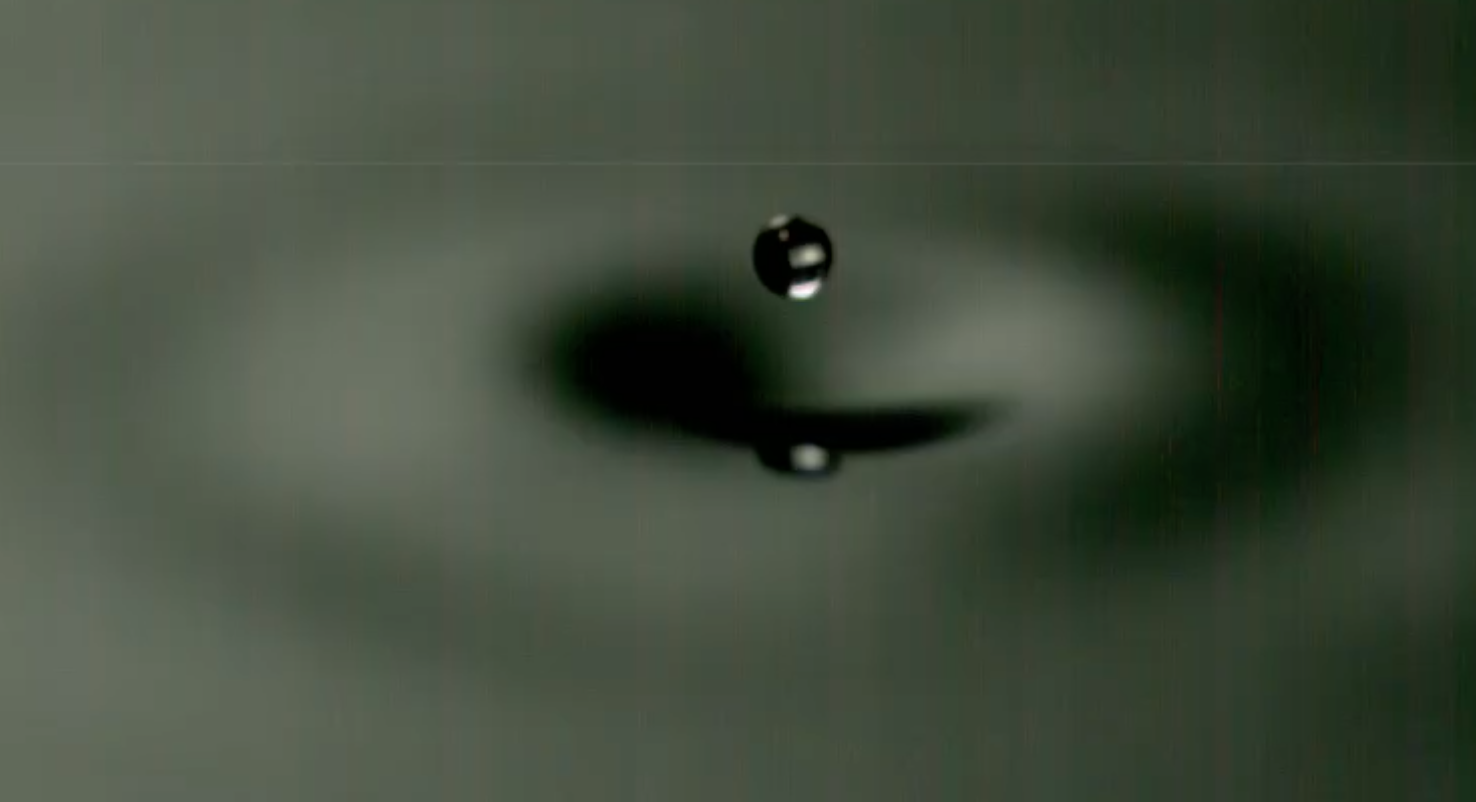}
\caption{Post-impact ascension.}
\end{subfigure}\hfill
    \caption{A millimetric bouncing droplet and its associated wave field, with the fluid bath bottom topography being a constant-depth circular corral. By setting the oscillatory acceleration $\gamma$ to be $1.6$ g and ensuring that the vibration number $\Omega$ --- which compares the angular forcing frequency to the droplet's oscillatory frequency --- satisfies $\Omega > 0.8$, we ensure that the droplet lies in the simple-bouncing regime of Figure \ref{fig-regime-diagram}.}
    \label{fig-wave-field-for-intro}
\end{figure}

\begin{figure}[t]
    \centering
    \includegraphics[width = 0.7 \columnwidth]{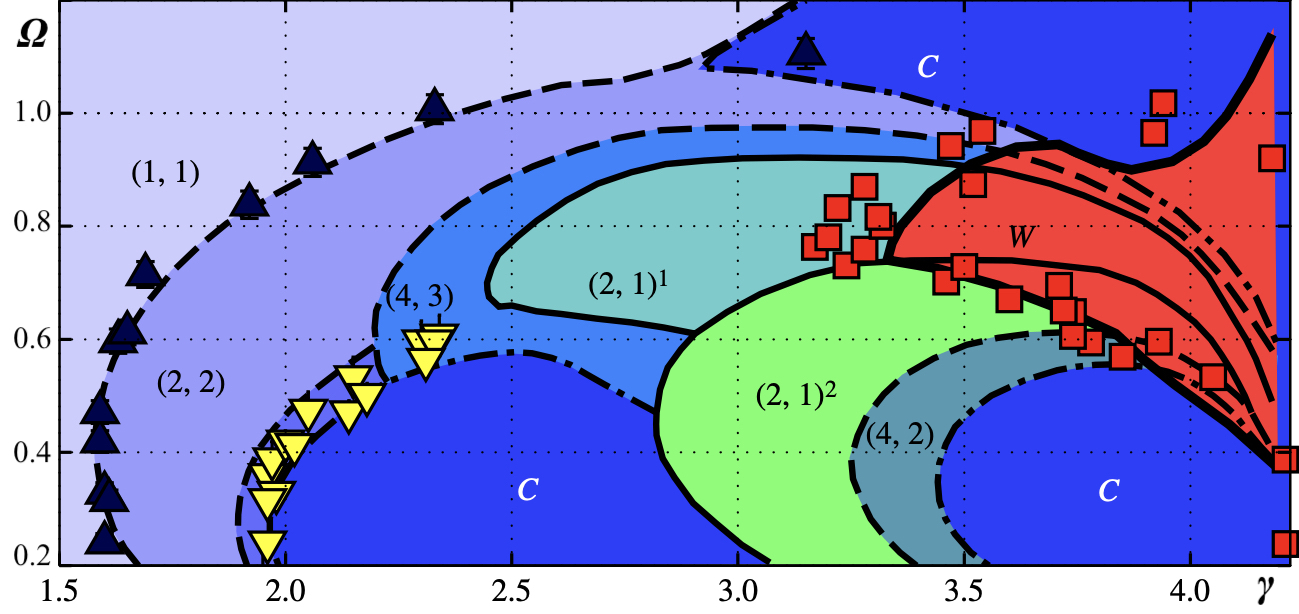}
    \caption{A phase diagram (adapted from \cite{molavcek2013dropsb}) at a driving oscillation frequency of 70 Hz, displaying the dependence of droplet dynamics on the forcing parameter $\Gamma = \gamma/\mathrm{g}$ (where $\gamma$ gives the magnitude of oscillatory forcing) and vibration number $\Omega$. The regions denoted by $C$ denote chaos, whereas the regions designated by $(a,b)$, for integers $a,b$, denote $b$ drop contacts occurring within $a$ forcing periods \cite{gilet2009chaotic}. Our analysis centers on the walking regime, denoted $W.$}
    \label{fig-regime-diagram}
\end{figure}

We will focus on the hydrodynamic quantum analogy of droplet tunneling between cavities \cite{nachbin2017tunneling, eddi2009unpredictable}, although our analytical models are more general. Our major contributions are threefold: (i) we develop a general theoretical characterization of three-dimensional droplet dynamics on a fluid bath (Section §\ref{sec3}), (ii) we conduct numerical simulations of droplet and wave-field time evolution, offering the closest known numerical approximation to experimental results (Section §\ref{sec4}), and (iii) we present experimental results on droplet tunneling in cavities with variable bottom topography (Section §\ref{sec5}). Prior to discussing such results, we describe the methodology and limitations of existing mathematical models for droplet dynamics (Section §\ref{sec2}).

\section{Existing Models of Droplet Dynamics}
\label{sec2}

Theoretical descriptions of walking droplet trajectories have gained significant recent interest \cite{oza2013trajectory, kaydureybush, bush2015pilot, molavcek2013drops,molavcek2013dropsb}. We discuss the advantages and limitations of well-known models characterizing droplet and wave-field behavior to motivate our extensions in Section §\ref{sec3}.

\paragraph{Stroboscopic Models.}

To establish a model that accurately describes the motion of walking droplets, we assume that our fluid lies below the Faraday instability threshold \cite{faraday1831xvii, benjamin1954stability}. Many prior models of droplet trajectories neglect the vertical component of the droplet's position under the \textit{stroboscopic approximation} \cite{molavcek2013drops}, which assumes that the droplet's vertical position may be simplified as periodic \cite{protiere2006particle}. Under such a stroboscopic model, Oza \textit{et al.}\ \cite{oza2013trajectory} utilize a force balance to derive a trajectory equation for the two-dimensional horizontal position function $\M{x}_p(t) = \langle x(t), y(t) \rangle$ for a walking droplet at time $t$:
\begin{align}
\label{eq_oza}
m\M{x}_p''(t) + D\M{x}_p'(t) = -F(t) \nabla h(\M{x}_p, t),
\end{align}
where $m$ is the droplet mass, $D$ is the viscous drag coefficient, $h$ is an explicitly constructed model of the wave-field at the current droplet position, and $F(t)$ is the external forcing on the drop \cite{molavcek2013drops}. Despite its analytical tractability, the stroboscopic model neglects variations in the bouncing period of the droplet; since we seek to analyze phenomena that exhibit large variations in oscillatory periods, we must venture beyond the stroboscopic approximation to model behaviors that more accurately reflect experimental conditions.



\paragraph{Droplet Tunneling Investigations: Existing Precedent.}

In order to obtain a precise characterization of the pilot-wave system, we must return to the fundamentals of fluid behavior. The underlying equations governing the motion of the a Newtonian fluid are the incompressible Navier-Stokes equations \cite{stokes1851effect, kundu2015fluid}, which we present as follows.
\begin{theorem}[Navier-Stokes \cite{stokes1851effect}]
    For an incompressible fluid with velocity $\M{u}(x, y, z, t)$, kinematic viscosity $\nu$, body acceleration $\M{F}(t)$, density $\rho$, and pressure $p(x, y, z, t)$, the time-derivative of velocity is given by the following partial differential equation:
\begin{align*}
    \M{u}_{t}  = - \frac{1}{\rho} \nabla p + \nu \Delta \M{u} + \M{F} - (\M{u} \cdot \nabla) \M{u},
\end{align*}
subject to the additional condition that $\nabla \cdot \M{u} = 0$ throughout the region of interest.
\end{theorem}

Using weakly viscous quasipotential theory and assuming low advective term as well as weak dissipation, the Navier-Stokes equations may be nondimensionalized, linearized, and simplified through a Helmholtz decomposition of the fluid velocity potential $\phi$, where $\phi$ satisfies $\nabla \phi = \M{u}$ throughout the fluid \cite{dias2008theory}. Previous studies apply this reduction to the pilot-wave system, condensing the system into a set of differential equations that encapsulate the complete dynamics of the three-dimensional wave-field \cite{milewski2015faraday}. The resulting equations describe the time-evolution of the velocity potential and the \textit{free surface} (the elevation of the fluid at the top of the bath). Assuming that the behavior of the droplet is invariant in one spatial coordinate, we may specify the details of evolution equations using the following theorem.

\begin{theorem}[Milewski, et al. \cite{milewski2015faraday}]
\label{milewski2015faraday-evolution-eta-phi}
    If $\eta(x, t)$ denotes the wave-field elevation in a fluid bath, $\phi(x, z, t)$ denotes the velocity potential at position $(x, 0, z)$ and time $t$, and $P_D(u)$ gives the pressure exerted by the droplet on the bath at position $u$, then the evolution of the wave-field due to droplet motion $x_p(t)$ may be modeled along the free surface as:
\begin{align}
\label{fs1}
    \frac{\ptl \phi(x, z, t)}{\ptl t} &= - g(t) \eta(x, t) + \frac{\sigma}{\rho} \frac{\ptl^2 \eta(x, t)}{\ptl x^2} + 2 \nu \frac{\ptl^2 \phi(x, z, t)}{\ptl x^2} - \frac{1}{\rho} P_D(x - x_p(t)), \\
\label{fs2}
    \frac{\ptl \eta}{\ptl t} &= \frac{\ptl \phi}{\ptl z} + 2 \nu \frac{\ptl^2 \eta}{\ptl x^2},
\end{align}
where $\rho$, $\sigma$, and $\nu$ denote the density, surface tension, and kinematic viscosity, respectively, of the fluid and $g(t)$ denotes oscillatory forcing.
\end{theorem}

The conditions specified by Theorem \ref{milewski2015faraday-evolution-eta-phi} necessitate a description of the droplet position $x_p(t)$, as every droplet impact is coupled with the wave-field itself. We adopt the formulation of Nachbin \textit{et al}.\ \cite{nachbin2017tunneling}; our trajectory equation mirrors the stroboscopic model in equation (\ref{eq_oza}), except that the forcing term $g(t)$ adds a vertical impact on droplet and free surface behavior in a linear spring model:
\begin{align}
\label{trajnew}
    m \cdot x_p''(t) + c F(t) \cdot x_p'(t) = -F(t) \cdot \eta_{x}(x_p(t), t),
\end{align}
where $c$ is damping constant. We impose two additional physical restrictions: there must be no fluid flow through the bottom of the fluid bath (that is,  $\phi_z = 0$ at $z = -H(\M{x})$, where $H(\M{x})$ specifies the depth of the bath at position $\M{x}$) and the velocity potential must satisfy Laplace's equation, $\Delta \phi = 0$, throughout the fluid bath. These constraints, together with boundary-value restrictions in Theorem \ref{milewski2015faraday-evolution-eta-phi} and equation (\ref{trajnew}), characterize droplet dynamics. 

It is computationally expensive to solve the system of differential equations (\ref{fs1}–\ref{trajnew}) on the entire fluid bath, due to the coupling of the equations and difficulties with approximations made by Computational Fluid Dynamics (CFD) methods. Accordingly, Nachbin \textit{et al.} \cite{nachbin2017tunneling} develop an alternative to the CFD approach by utilizing a two-dimensional Dirichlet-to-Neumann (DtN) operator to recover the evolution of the velocity potential $\phi$ from wave-field evolution at the fluid interface. Accordingly, we present an extension of the DtN operator into three dimensions in the following section, which allows for a tight approximation to physical droplet trajectories.

\section{A Generalized Hydrodynamic Model for Droplet Dynamics}
\label{sec3}

In 2009, Eddi \textit{et al.}\ \cite{eddi2009unpredictable} observed experimental droplet tunneling; that is, a walking droplet traversed a subsurface barrier between two deeper wells. Nachbin \textit{et al.}\ \cite{nachbin2017tunneling} established a numerical model for droplet trajectories under variable topography in 2017, which could then be applied to investigate tunneling --- assuming invariance of fluid behavior in one spatial dimension. Accordingly, we develop a model for complete three-dimensional droplet dynamics under variable bottom topography without assumption of invariance.

\subsection{Distilling Droplet Dynamics in Three Dimensions}

From the reductions invoked to obtain the result of Theorem \ref{milewski2015faraday-evolution-eta-phi}, it follows that every observable property of our fluid-dynamical system may be obtained from the scalar-valued functions $\phi$, $\eta$, assuming a one-dimensional droplet position function $x_p(t)$. Motivated by the extensions developed by Faria in 2017 \cite{faria2017model} and the lower-dimensional model of Nachbin \textit{et al.} \cite{nachbin2017tunneling}, we offer an explicit characterization of the evolution of $\phi$ and $\eta$ in the case where the droplet position $\M{x}_p(t)$ varies in both horizontal components.
\begin{theorem}
For a velocity potential $\phi(x, y, z, t)$ in a compact fluid region $\mathcal{R}$ with bottom topography $H(\M{x})$ and free surface $\eta$ initially stable at $z = 0$, the following equations characterizing the spatiotemporal evolution of the wave-field and velocity potential hold:
    \begin{align}
\centering
\label{thm_3_first_first_eq}
    0 &= \nabla \phi, \\
\label{bottom_thm_3.1}
    0 &= \phi_z - \phi_x H_x - \phi_y H_y,\\
\label{phi_t_thm_3.1}
    \phi_t &= -g(t) \eta + \frac{\sigma}{\rho} \Delta_{H} \eta  + 2 \nu \Delta_{H} \phi - \frac{1}{\rho} P_D(\M{x} - \M{x}_p(t), t) ,\\
\label{eta_t_thm_3.1}
    \eta_t &= \phi_z + 2 \nu \Delta_{H} \eta,
\end{align}
where $\Delta_{H}$ denotes the horizontal Laplacian operator (i.e., $\Delta_H = \partial_x^2 + \partial_y^2$). The first condition holds throughout the fluid, the second holds at the bottom of the bath, and the third and fourth hold at the free surface (Figure \ref{fig-intuition}).
\end{theorem}
\begin{proof}[Proof Sketch]
    The result is a direct extension of the derivation of analogous equations by Milewski \cite{milewski2015faraday}; in particular, all four equations result from nondimensionalizing weakly viscous equations for water-wave dynamics. Accordingly, we omit the proof.
\end{proof}

In equations (\ref{thm_3_first_first_eq}–\ref{eta_t_thm_3.1}), several parameters require concrete specification. Following precedent in numerical methods, we assume that the forcing oscillation $g(t)$ in equation (\ref{phi_t_thm_3.1}) varies with amplitude equal to the forcing amplitude $\gamma$; that is, $g(t) = \mathrm{g} + \gamma \cos(\omega t)$ for oscillatory frequency $\omega$ \cite{oza2013trajectory, eddi2009unpredictable}. We discuss two possible approximations for the more complex droplet pressure term, $P_D$, in Section §\ref{sec4}. We now invoke a nondimensionalization procedure to transform our system into one involving physically motivated parameters (see Figure \ref{fig-intuition}). 

\begin{theorem}
\label{nondim_procedure}
The evolution of the wave-field and a walking droplet on a fluid bath with variable topography can be characterized through the following differential equations:
\begin{align}
    \label{laplacian_1}
    0 &= \mu^2 \Delta_{{H}}{{\phi}} + \phi_{zz}, \\
\label{bottom_1}
    0 &= \phi_z - \mu^2({{\phi}_{{x}} {H}_{{x}} + {\phi}_{{y}} {H}_{{y}}}),\\
\label{weird_1}
    \phi_t &= - G \cdot  (1 + \gamma \cos(\omega t)) {\eta} + \frac{2}{\mathrm{Re}}  \Delta_{{H}}{{\phi}} + \mathrm{Bo} \, \Delta_{{H}}{{\eta}} - \frac{GM}{\rho} P_D({\M{{x}}} - {\M{{x}}_p(t)}, t), \\
\label{kinematical_condition_1}
    {\eta}_{{t}} &=  \frac{1} {\mu} {\phi}_{{z}} +\frac{2}{\mathrm{Re}} \Delta_H{{\eta}}, \\
\label{position_update}
    - F(t) \nabla \eta (\M{x}_p) &=  m \M{x}_p'' + c F(t) \M{x}_p'.
\end{align}
with\footnote{For completeness, we note that $T_F$ designates the Faraday period, and $\lambda_F$ designates the Faraday wavelength.} $\mu = \frac{h}{\lambda_F}, G = \frac{\text{g} T_F^2}{\lambda_F}$, $\mathrm{Bo} = \frac{\sigma T_F^2}{\rho \lambda_F^3}$, $M = \frac{m}{\rho \lambda_F^3}$.
\end{theorem}

\begin{proof}[Proof.] We introduce the following nondimensionalization schema:
\begin{align*}
    \Bigg \{
    \begin{array}{ll}
      \til{x} = \ell x, \, \, \, \til{y} = \ell y, \, \, \, \til{z} = hy, \, \, \, \\
    \til{t} = t_F t, \, \, \, \til{\eta} = \eta_c \eta, \, \, \, \til{\phi} = \phi_c \phi,\til{H} = H_c H. 
    \end{array}
\end{align*}
 
where $H_c, \eta_c, \ell$ are constants and $h, t_F$ give the mean depth and Faraday period. We now apply this schema to equations (\ref{thm_3_first_first_eq}–\ref{eta_t_thm_3.1}). For equation (\ref{thm_3_first_first_eq}), we have:
\begin{align*}
    \frac{\ptl}{\ptl x} \left(\phi_c \frac{\ptl \til{\phi}}{\ptl \til{x}} \frac{1}{\ell} \right) + \frac{\ptl}{\ptl y} \left(\phi_c \frac{\ptl \til{\phi}}{\ptl \til{y}} \frac{1}{\ell} \right) + \frac{\ptl}{\ptl z} \left(\phi_c \frac{\ptl \til{\phi}}{\ptl \til{z}} \frac{1}{h} \right) = 0.
\end{align*}

Simplification and rearrangement in accordance with the horizontal Laplacian gives the following elliptic partial differential equation, with $\mu := h/\ell$ being a nondimensional parameter:
\begin{align}
\label{nondim-laplace}
    \mu^2 \Delta_{\til{H}}{\til{\phi}} + \frac{\ptl^2 \til{\phi}}{\ptl \til{z}^2} = 0.
\end{align}

Similarly, we may express the nondimensionalized form of equation (\ref{bottom_thm_3.1}) as follows:
\begin{align*}
    \phi_c \frac{\ptl \til{\phi}}{\ptl \til{z}} \frac{1}{z_c} - \phi_c  \frac{\ptl \til{\phi}}{\ptl \til{x}} \frac{1}{\ell} \cdot H_c \frac{\ptl \til{H}}{\ptl \til{x}} \frac{1}{\ell}  - \phi_c  \frac{\ptl \til{\phi}}{\ptl \til{y}} \frac{1}{\ell} \cdot H_c \frac{\ptl \til{H}}{\ptl \til{y}} \frac{1}{\ell} = 0.
\end{align*}
Applying our substitution for the nondimensional parameter $\mu$ gives the following:
\begin{align}
\label{nondim-bottom}
    \frac{\ptl \til{\phi}}{\ptl \til{z}} - \mu^2{\til{\phi}_{\til{x}} \til{H}_{\til{x}} + \til{\phi}_{\til{y}} \til{H}_{\til{y}}} = 0.
\end{align}

Next, we nondimensionalize the stress condition (\ref{phi_t_thm_3.1}) as follows:
\begin{align*}
    \phi_c \frac{\ptl \til{\phi}}{\ptl \til{t}} \frac{1}{t_F} &= -g(t) \eta_{c} \til{\eta} + \frac{\sigma}{\rho} \cdot \frac{\eta_c^2}{x_c^2} \left( \frac{\ptl^2 \til{\eta}}{\ptl \til{x}^2} + \frac{\ptl^2 \til{\eta}}{\ptl \til{y}^2} \right) \\
    & \, \, \, + 2 \nu \frac{\phi_c^2}{x_c^2} \Delta_{\til{H}}{\til{\phi}} - \frac{1}{\rho} P_D(x_c ({\M{\til{x}}} - {\M{\til{x}}_p(t)}), t_c \til{t}).
\end{align*}
After including the Reynolds number Re := $\frac{2 \nu t_F}{\lambda_F^2}$ for our fluid, we may write:
\begin{align}
\label{eqn33}
    \frac{\phi_c}{t_F} \frac{\ptl \til{\phi}}{\ptl \til{t}} &= -g(t) \eta_{c} \til{\eta} + \mathrm{Re}  \Delta_{\til{H}}{\til{\nu}} + 2 \nu \frac{\phi_c^2}{x_c^2} \Delta_{\til{H}}{\til{\phi}} - \frac{1}{\rho} P_D(x_c ({\M{\til{x}}} - {\M{\til{x}}_p(t)}), T_F \til{t}).
\end{align}

Finally, the analogue of the kinematic condition (\ref{eta_t_thm_3.1}) in the nondimensional case is:
\begin{align}
\label{eqn44}
    \til{\eta}_{\til{t}} = \frac{t_F \phi_c}{\eta_c h} \til{\phi}_{\til{z}} + \mathrm{Re} \nabla_{\til{H}}{\til{\eta}}.
\end{align}
Equations (\ref{nondim-laplace}-\ref{eqn44}) specify the nondimensionalized version of our system.
\end{proof}

\begin{figure}[!t]
    \centering
    \includegraphics[width=0.6\linewidth]{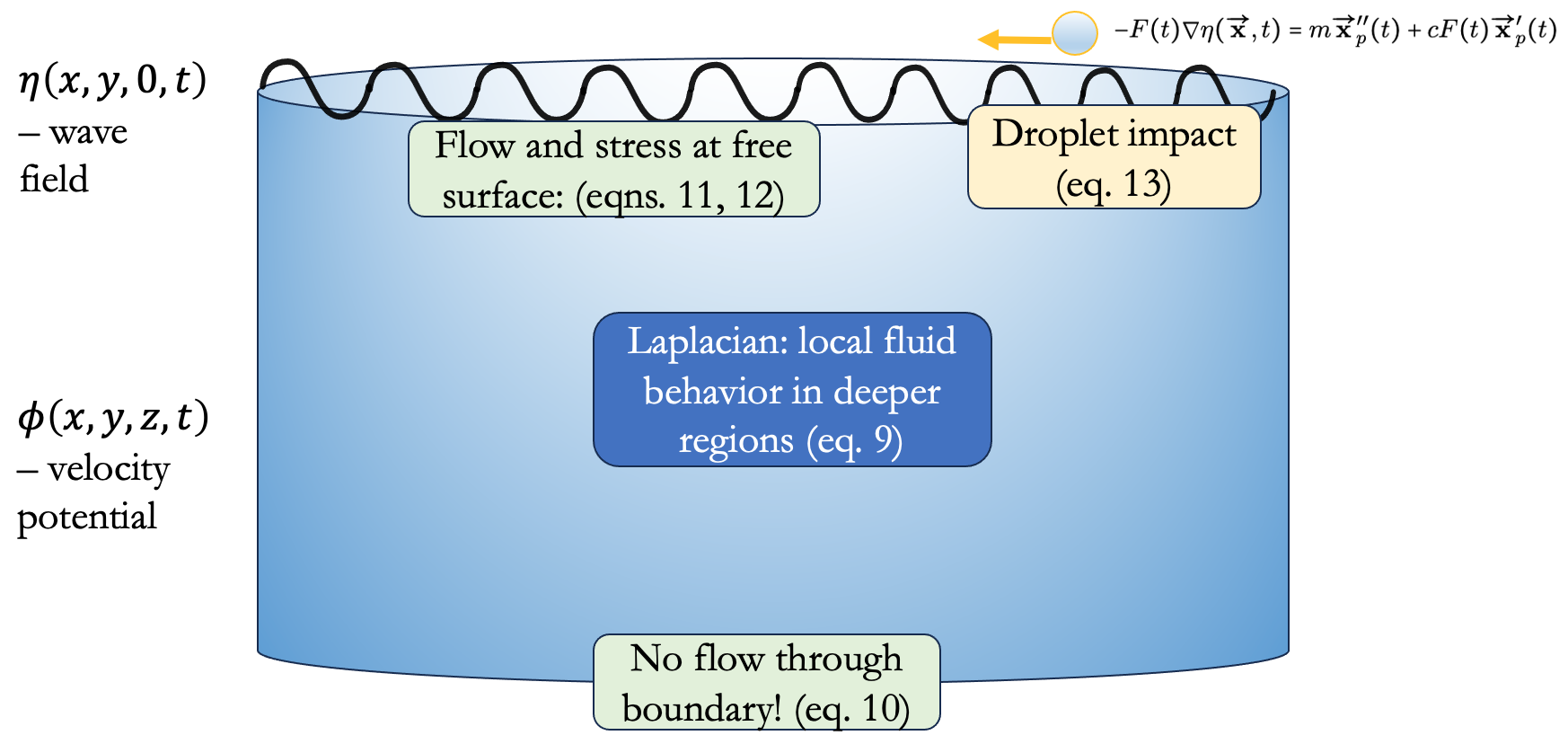}
    \caption{Schematic of the constraints on the elliptic PDE (\ref{laplacian_1}) in Theorem \ref{nondim_procedure}, which describe droplet dynamics on a fluid bath with variable bottom topography in three spatial dimensions. Equations (\ref{bottom_1}–\ref{kinematical_condition_1}) describe the conditions for fluid flow along the boundaries of the region, while Equation (\ref{position_update}) describes the droplet trajectory itself.}
    \label{fig-intuition}
\end{figure}

In order to solve the elliptic system of partial differential equations (PDE's) in Theorem \ref{nondim_procedure} at the free surface, we must characterize the vertical behavior of the velocity potential, namely $\phi_z$. Specifically, we require $\phi_z$ to evolve the wave-field $\eta$, using equation (\ref{kinematical_condition_1}), and the velocity potential, using equation (\ref{weird_1}), during droplet dynamics. In the following section, we approximate $\phi_z$ using a Dirichlet-to-Neumann (DtN) operator.

\subsection{Constructing a Three-Dimensional Dirichlet-to-Neumann Operator}

We seek to establish an operator, denoted DtN, that transforms \textit{Dirichlet boundary data} --- i.e. explicit values $q(\M{x}, 0)$ of the potential $\phi(\M{x}, z)$ on the free surface --- into \textit{Neumann boundary data} --- i.e., the values of $\phi_z(\M{x}, z)$. We first consider the case in which the bath has a constant depth. From the boundary values of $\phi$ on the free surface, we use Fourier analysis to explicitly determine the values of $\phi_z$ on the free surface, as follows.
\begin{theorem}
\label{flat_bottom_DtN}
    Assuming initial values $\phi(\M{x}, 0) = q(\M{x})$ on the free surface of a fluid bath and flat-bottom topography on the region $[0, L]^2$ for some $L \in \mathbb{R}$, we have:
    \begin{align}
    \phi_z(\M{x}, 0) = \mathrm{DtN}[q] = \sum_{k \in \Lambda^*} \exp(\mathrm{i} \M{k} \cdot \M{x}) \cdot \hat{q}(\M{k}) \mu k \tanh(\mu k),
\end{align}
where $\Lambda^*$ denotes the set of Fourier wavevectors $\mathbf{k}$ and $k = ||\M{k}||_2$.
\end{theorem}

\begin{proof}[Proof.]
We will solve the elliptic PDE (\ref{laplacian_1}) by a Fourier transform method, which allows us to determine the potential $\phi$ and its partial derivative $\phi_z$ in Fourier space.

Since equation (\ref{laplacian_1}) is autonomous in $t$, we temporarily omit consideration of temporal behavior and perform an two-dimensional Fourier transform $\mathcal{F}$ of both sides of the equation. We may express the transformed differential equation as follows\footnote{We define the two-dimensional Fourier transform as follows: if $f \in L^1(\R)$, where $L^1$ denotes Lebesgue space, then we write $\mathcal{F}[f](\M{k}) := \int_{\R} \int_{\R} \exp(2 \pi i \M{k} \cdot \M{x}) f(\M{x}) \,  \mathrm{d}\M{x}_1 \, \mathrm{d}\M{x}_2$.}:
\begin{align}
\label{int_const_depth_1}
    \mathcal{F}[\mu^2 \Delta_H \phi + \phi_{zz}] = \mathcal{F}[0] = 0.
\end{align}
By linearity of the Fourier transform, we may expand the left-hand side as $\mu^2 (\mathcal{F}[\phi_{xx}] + \mathcal{F}[\phi_{yy}]) + \mathcal{F}[\phi_{zz}]$. To simplify, we note that, for a differentiable function $\alpha \in L^1(\R)$, we have that $\mathcal{F}[D \alpha](\M{k}) = 2 \pi i \M{k} \mathcal{F}[\alpha]$, where $\M{k}$ is an element of Fourier space and $D$ is the derivative operator. Applying this result twice to the left-hand side of equation (\ref{int_const_depth_1}) gives:
\begin{align*}
    \mu^2 (\mathcal{F}[\phi_{xx}] + \mathcal{F}[\phi_{yy}]) + \mathcal{F}[\phi_{zz}] = -\mu^2 \left(\left(2 \pi i \M{k} \right)^2 \hat{\phi} + \left(2 \pi i \M{k} \right) ^2  \hat{\phi} \right) + \hat{\phi}_{zz} = - \mu^2 ||\M{k}||_2^2  \hat{\phi} + \hat{\phi}_{zz}.
\end{align*}
Our transformed ordinary differential equation (ODE) is autonomous in $x$ and $y$; therefore, we have reduced a PDE in three spatial variables into an ODE in $z$. Solving this equation by standard methods for second-order, linear ODE's gives:
\begin{align}
\label{guess_velocity_potential_const_depth}
    \hat{\phi}(\M{k}, z) = A(\M{k}) \cosh(\mu ||\M{k}||_2^2 \cdot  z) + B(\M{k}) \sinh(\mu ||\M{k}||_2^2 \cdot  z),
\end{align}
where $A(\M{k}), B(\M{k})$ are undetermined coefficients. We now use our boundary conditions to determine the coefficients $A(\M{k}), B(\M{k})$. Since the fluid depth $H(\M{x})$ is a constant $h_0$, equation (\ref{bottom_1}) gives $\phi_z = 0$ at $z = -h_0$. Substituting this result into equation (\ref{guess_velocity_potential_const_depth}) yields:
\begin{align}
\label{int_const_depth_2}
     A(\M{k}) \sinh(\mu k h_0) = B(\M{k}) \cosh(\mu k h_0).
\end{align}
If we assume knowledge of Dirichlet data $\hat{q}(\M{k})$ at the free surface, then we may evaluate the left-hand side of equation (\ref{guess_velocity_potential_const_depth}) at $z = 0$ to describe the coefficients of equation (\ref{int_const_depth_2}):
\begin{align}
\label{int_const_depth_3}
    A(\M{k}) = \hat{q}(\M{k}), \\
\label{int_const_depth_4}
     B(\M{k}) = \frac{\hat{q}(\M{k}) \cdot \sinh(\mu k h_0)}{\cosh(\mu k h_0)}.
\end{align}
Therefore, we may derive an explicit expression for $\hat{\phi}(\M{k}, z)$ from equations (\ref{guess_velocity_potential_const_depth}), (\ref{int_const_depth_3}), and (\ref{int_const_depth_4}), as follows:
\begin{align}
\label{vel_pot_almost}
    \hat{\phi}(\M{k}, z) &= \hat{q}(\M{k}) \cosh(-\mu k h_0) - \frac{\hat{q}(\M{k}) \sinh(-\mu k h_0)}{\cosh(-\mu k h_0)} 
= \hat{q}(\M{k}) \frac{\cosh(\mu k (-h_0 + 1) )}{\cosh(\mu k)}.
\end{align}
Finally, computing a derivative with respect to $z$ and performing an inverse Fourier transform gives a physical characterization of the Dirichlet-to-Neumann operator for the constant depth case. We specify the resulting DtN operator as:
\begin{align*}
    \text{DtN}[q] := \phi_z(\M{x}, 0) = \sum_{\M{k} \in \Lambda^*} \exp(i \M{k} \cdot \M{x}) \cdot \hat{q}(\M{k}) \mu k \tanh(\mu k),
\end{align*}
where $\Lambda^* = \frac{2\pi}{L} \Z \times \frac{2\pi}{L} \Z - \{0\}$ is the set of nonzero wave-vectors $\M{k}$ in Fourier space. 
\end{proof}
    From Theorem \ref{flat_bottom_DtN}, it is also possible to derive an equation for the velocity potential through a direct application of the Fourier transform, rather than differentiating with respect to $z$. Doing so gives the following equation for the velocity potential $\phi(x, y, z)$:
\begin{align}
\label{vel_pot}
    \phi(x, y, z) = \hat{q}(0) + \sum_{\M{k} \in \Lambda^*} \exp(i \M{k} \cdot \M{x}) \hat{q}(\M{k}) \frac{\cosh(\mu k (z + 1) )}{\cosh(\mu k)}.
\end{align}

We now consider the general case involving variable bottom topography of the fluid bath. Following the incorporation of the bottom topography by Milewski \cite{milewski1998formulation}, we add an extra term to the velocity potential of equation (\ref{vel_pot}) due to variations in the depth $H(\M{x})$:
\begin{align}
\label{3d-dtn-3.1}
      \phi(x, y, z) = \hat{q}(0) + \sum_{\M{k} \in \Lambda^*} \exp(i \M{k} \cdot \M{x}) \left[ \hat{q}(\M{k}) \frac{\cosh(\mu k (z + 1) )}{\cosh(\mu k)} + X_{\M{k}} \frac{\sinh(\mu kz)}{k \cosh^2(\mu k)} \right],
\end{align}
where the sequence of \textit{topographic coefficients} $(X_{\M{k}})_{\M{k} \in \Lambda^*}$ characterize the bottom topography of the fluid bath. This expression is also contained within Milewski \cite{milewski1998formulation}. We may relate the unknown coefficients $(X_{\M{k}})$ to the known Dirichlet data values $q(x, y)$ of the potential on the free surface by imposing the boundary condition in equation (\ref{bottom_1}). Equating the resulting terms gives the following:
\begin{align}
\label{3d-dtn-3.2}
    \sum_{\M{k} \in \Lambda^*} \hat{q}(\M{k}) \nabla \cdot \left[ \mathrm{e}^{i \M{k} \cdot \M{x}}  \frac{\sinh(\mu k H(\M{x}) )}{\cosh(\mu k)} \frac{\M{k}}{k} \right] &= \sum_{\M{k} \in \Lambda^*} X_{\M{k}} \nabla \cdot \left[ \mathrm{e}^{i \M{k} \cdot \M{x}} \frac{\cosh(\mu k (1 + H(\M{x}) ))}{\cosh^2(\mu k)} \frac{\M{k}}{k^2} \right],
\end{align}
Note that the term on the left-hand side of the above equation originates from the formulation of the Dirichlet-to-Neumann operator in the flat-bottom case, whereas the term on the right-hand side involves the variable topography term.

Once the topographical coefficients $X_{\M{k}}$ have been calculated from equation (\ref{3d-dtn-3.2}), we may use these coefficients to construct the full Dirichlet-to-Neumann (DtN) operator in three dimensions by differentiating equation (\ref{3d-dtn-3.1}) with respect to $z$:
\begin{align}
\label{3d-dtn-3.3}
     \text{DtN}[q] := \phi_z(\M{x}, 0) = \sum_{k \in \Lambda^*} \exp(i \M{k} \cdot \M{x}) \cdot  \left[ \hat{q}(\M{k}) \cdot \mu k \tanh(\mu k) + \mu X_{\M{k}}  \cdot \text{sech}^2(\mu k) \right].
\end{align}

Therefore, we may derive all information regarding the behavior of the fluid \textit{throughout} our system solely from the \textit{surface} boundary values, $q$, and the topographic constraints of the bath, encapsulated by $X_{\M{k}}$. We address the calculation of the topographical coefficients by means of equation (\ref{3d-dtn-3.2}) in the following section.

\subsection{Calculation of Topographical Coefficients through Galerkin Method.}

To obtain the values of the coefficients $X_{\M{k}}$ from the formulation presented in equation (\ref{3d-dtn-3.2}), we establish two operators to characterize the left- and right-hand sides, following the methodology of Andrade \textit{et al.}\ \cite{andrade2018three}. We define the action of the operators $A$ and $B$ on the sequences $\hat{q}(\M{k})$ and $(X_{\M{k}})$ as follows:
\begin{align*}
    A[\hat{q}(\M{k})] &:= \sum_{\M{k} \in \Lambda^*} \hat{q}(\M{k}) \nabla \cdot \left[ \mathrm{e}^{i \M{k} \cdot \M{x}}  \frac{\sinh(\mu k H(\M{x}) )}{\cosh(\mu k)} \frac{\M{k}}{k} \right], \\
    B[X_{\M{k}}] &= \sum_{\M{k} \in \Lambda^*} X_{\M{k}} \nabla \cdot \left[ \mathrm{e}^{i \M{k} \cdot \M{x}} \frac{\cosh(\mu k (1 + H(\M{x}) ))}{\cosh^2(\mu k)} \frac{\M{k}}{k^2} \right].
\end{align*}
Then, equation (\ref{3d-dtn-3.2}) amounts to solving $A[\hat{q}({\M{k}})] = B[X_{\M{k}}]$ for the coefficients $X_{\M{k}}$. Consider a truncated sum of the coefficients in the expansion of $B[X_{\M{k}}]$, which bounds the magnitude $k$ of a wave-vector $\M{k}$ using the Galerkin parameter $M$: 
\begin{align*}
    B_M [X_{\M{k}}] = \sum_{\substack{k \leq M \\ k \in \Lambda^*}} X_{\M{k}} \nabla \cdot \left[ \mathrm{e}^{i \M{k} \cdot \M{x}} \frac{\cosh(\mu k (1 + H(\M{x}) ))}{\cosh^2(\mu k)} \frac{\M{k}}{k^2} \right].
\end{align*}

Solving our system amounts to minimizing the residual error in our estimate, which is $R[X_{\M{k}}] := B_M [X_{\M{k}}] - A[\hat{q}({\M{k}})]$. If the residual itself is orthogonal to the vector space of Fourier vectors $V_M := \text{span}\{\exp(i \M{k} \cdot \M{x}) : k \leq M\}$, then we have minimized the error in our approximation. Therefore, we must impose the condition\footnote{The notation $\langle f, g \rangle$, for integrable functions $f, g$, denotes the inner product of the functions $f, g$ in Fourier space.}
\begin{align}
\label{last-eqn}
    \langle B_M[X_{\M{k}}],\mathrm{e}^{i \M{w} \cdot \M{x}} \rangle = \langle A[\hat{q}({\M{k}})],\mathrm{e}^{i \M{w} \cdot \M{x}} \rangle,
\end{align} 
for some basis vector $\mathrm{e}^{i \M{w} \cdot \M{x}} \in V_M$. From the system of ordinary differential equations (ODE's) imposed by (\ref{last-eqn}), we may deduce the values of the coefficients $X_{\M{k}}$ numerically, which leads to a full solution of equation (\ref{3d-dtn-3.3}) for the behavior of the potential $\phi_z$ across the free surface. Aided by this description of $\phi_z$, we may use equations (\ref{weird_1}) and (\ref{kinematical_condition_1}) to specify the complete time-evolution of our system. Hence, we have effectively reduced a three-dimensional elliptic PDE system to a set of tractable, one-dimensional ODE's. We apply the resulting DtN map in the numerical implementation of our droplet model in Section §\ref{sec4}.

\section{Numerical Simulations}
\label{sec4}
We utilize the theoretical framework established in Section §\ref{sec3} to conduct three-dimensional simulations of walking droplet dynamics in MATLAB over variable-topography systems with limited computational expense. 
\subsection{Methodology and Simulation Framework}
\label{sec4.1}
We evolve the underlying dynamical system by direct time marching of the wave-field $\eta$, velocity potential $\phi$, and droplet position $\M{x}_p$, according to the description established in Theorem \ref{nondim_procedure}.

\begin{figure}[!htb]
\begin{subfigure}[t]{0.5\textwidth}
\centering
  \includegraphics[width = \textwidth]{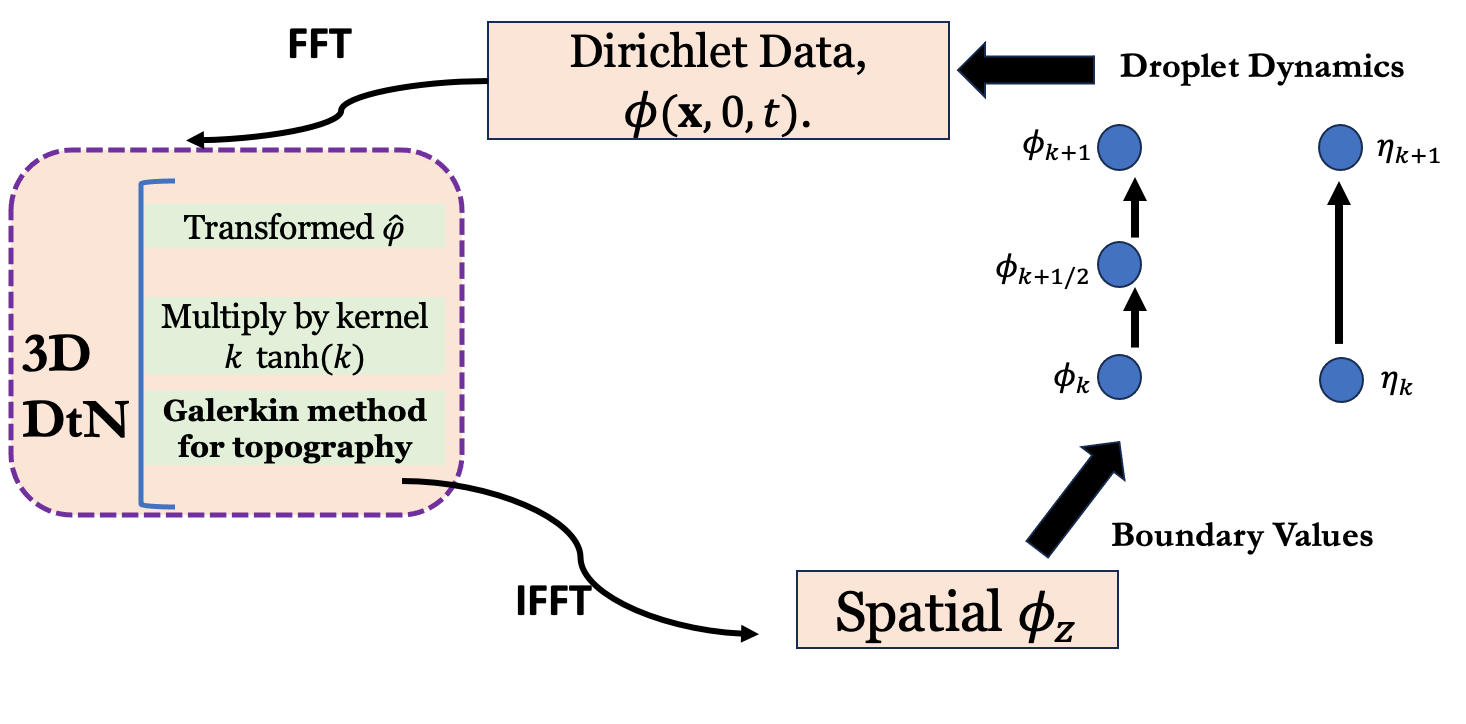}
\caption{Central-difference method: workflow. \label{centraldifference}
}
\end{subfigure}\hfill
\begin{subfigure}[t]{0.5\textwidth}
\centering
    \includegraphics[width = \textwidth]{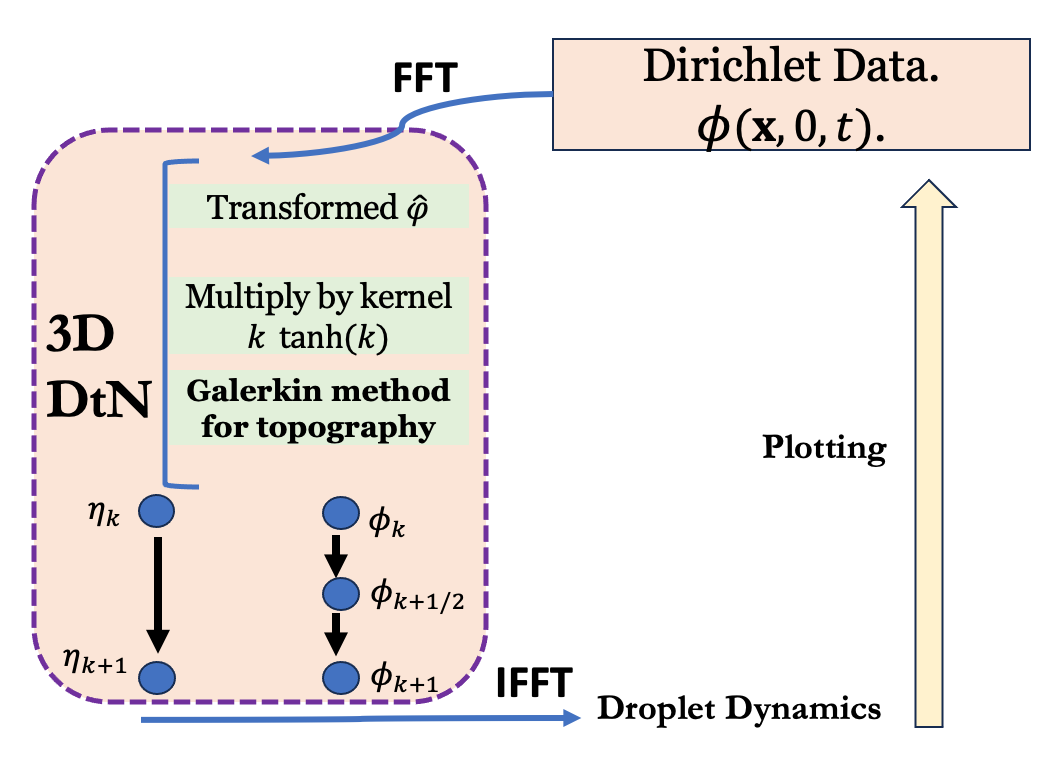}
\caption{Fourier pseudo-spectral method: workflow. \label{pseudospectral}
}
\end{subfigure}
    \caption{Workflow for analysis of droplet dynamics. In Figure \ref{centraldifference}, we demonstrate central-difference schema workflow, which requires frequent Fourier transforms (FFT and IFFT denote the Fast Fourier Transform and Inverse Fast Fourier Transform, respectively) The Fourier pseudo-spectral method in Figure \ref{pseudospectral} provides a more streamlined approach. Created by student researcher.}
    \label{fig6}
\end{figure}

The central feature of our numerical model is the usage of the Dirichlet-to-Neumann operator formulation in equation (\ref{3d-dtn-3.3}) to evolve the wave-field in equation (\ref{eta_t_thm_3.1}). In Figure \ref{fig6}, we provide a schematic of the efficiency of our surrogate three-dimensional DtN map, which reduces the computational complexity of our implementation. We note, however, that despite the simplified description given by Theorem \ref{nondim_procedure}, there remains to be specification of the droplet pressure term $P_D$ and the computation of the second partial derivatives of our potential and wave-field. We now address both of these issues. 

First, we must explicitly designate the droplet pressure term $P_D(\M{x} - \M{x}_p, t)$ analytically. Approximating the droplet trajectory by a standard linear spring model, which has been supported from both experimental \cite{molavcek2013dropsb} and analytical \cite{nachbin2017tunneling} standpoints, we assume instantaneous contact and, therefore, set $P_D(\M{x} - \M{x}_p, t) = F(t) \delta(\M{x} - \M{x}_p)$, where $\delta$ is the Dirac delta distribution and $F(t)$ is a forcing term given by
\begin{align}
\label{pressure_term}
    F(t) := \frac{8 \pi^2}{\omega T_F} \cdot \sin(\frac{4 \pi \tau(t)}{T_F}) \delta(\M{x} - \M{x}_p).
\end{align}
In the above formulation, we take the periodic indicator $\tau(t)$ to be given by $\tau(t) := t \pmod{T_F}$ for $0 \leq \tau(t) < T_F/4$; in effect, we take the droplet to be in contact with the wave-field for one-fourth of a Faraday period and assume that droplet dynamics are decoupled from the wave-field for the other three-fourths of the period --- an approximation that is viable for most tunneling scenarios under fluctuation of oscillatory periods.

Second, we compute the second partial derivatives of $\phi, \eta$ using two distinct approaches: (i) a second-order central difference method for computation of the relevant horizontal Laplacian operators, and (ii) a Fourier pseudo-spectral method that avoids the calculation of second derivatives in estimating the Laplacian.

The central difference method (Figure \ref{centraldifference}) involves simply evolving the equations in Theorem \ref{nondim_procedure} using the Dirichlet-to-Neumann operator and estimating the values of $\Delta_{H}\eta, \Delta_H{\phi}$ using a second-order central difference approximation on a two-dimensional mesh. However, this approach requires applying a Fourier inversion map during each time-evolution step. We may simplify the computational procedure by instead evolving our equations through a pseudo-spectral method (Figure \ref{pseudospectral}). To accomplish this reduction, we must implement an effective method to evolve the equations in Theorem \ref{nondim_procedure} in Fourier space.


We formulate our pseudo-spectral method by applying Fourier transforms to the restrictions in Theorem \ref{nondim_procedure}, which produces the following nondimensionalized evolution equations for the wave-field, velocity potential, and droplet velocity involving Fourier wave-vectors $\M{k}$:
\begin{align}
\label{ps_eq_1}
    \hat{\phi}_t &= -G(1 + \gamma \cos(4 \pi t)) \hat{\eta} - 2 \frac{||\M{k}||_2^2}{\text{Re}} \cdot  \hat{\phi} - ||\M{k}||_2^2 \cdot \text{Bo} \cdot \hat{\eta} - GM F(t) \exp(\mathrm{i} \M{k} \cdot \M{x}_p), \\
\label{ps_eq_2}
    \hat{\eta}_t &= -\frac{1}{\mu}\hat{\phi}_z - \frac{2}{\text{Re}} ||\M{k}||_2^2 \hat{\eta}, \\
\label{ps_eq_3}
    \M{v}_t &= \frac{-8 T_F \mathrm{g} \pi^2}{\lambda_F \omega} \nabla \eta - \frac{8 c g \pi^2}{\omega} \sin\left(\frac{4\pi \tau(t)}{T_F}\right) \M{v}.
\end{align}

\begin{algorithm}[!t]
\caption{Condensed Description of Pseudo-Spectral Evolution Methodology}\label{alg1}
\hspace*{\algorithmicindent} \textbf{Input:} Parameters ($\gamma, \text{Re}, \text{Bo}, G, M, c, T_F, \lambda_F, \omega, {\Delta x}, {\Delta y},\Delta t$, $\Delta t'$); Fourier wave-vectors $\M{k}$; as well as initial wave-field, $\eta$,  velocity potential, $\phi$, and droplet trajectory, $\M{x}_p, \M{v}$.\\
 \hspace*{\algorithmicindent} \textbf{Output:} Final wave-field $\eta$ and droplet position $\M{x}_p$.
\begin{algorithmic}[1]
    \State{$\hat{\phi}, \hat{\eta} \gets \text{FFT}(\phi), \text{FFT}(\eta)$} \Comment{Fourier Transform of initial values.}
    \While{$t \leq t_{\text{max}}$}
        \While{$t' \leq 1$} \Comment{Smaller time-stepping within each Faraday period.}
           \If{$t' \leq 1/4$} \Comment{Droplet is in contact with bath.}
           \State{$\hat{\phi} \gets \hat{\phi} + \frac{\Delta t}{2} \cdot \text{Evolve}(\hat{\phi})$} \Comment{Half step in $\hat{\phi}$ using eq.\ (\ref{ps_eq_1}).}
           \State{$\M{v} \gets \M{v} + \frac{\Delta t}{2} \cdot \text{Evolve}(\M{v})$}\Comment{Update $\M{v}$ using eq.\ (\ref{ps_eq_3}).}
            \State{$\widehat{\phi_z} \gets \text{DtN}[\hat{\phi}]$}\Comment{Implement DtN operator}
            \State{$\M{x}_p \gets \M{x}_p + \Delta t \cdot \M{v}$}\Comment{Update droplet position.}
            \State{Repeat lines 8 – 9 for half-step in $\phi$.}
            \Else \Comment{Droplet is not in contact with bath.}
            \State{Execute only lines 8, 10 – 12, 13.}\Comment{Velocity is invariant while not contact.}   
           \EndIf
            
        \State{$t' \gets t + \Delta t'$}\Comment{Increment smaller time-step.}
        \EndWhile
   
    \State{$t \gets t + \Delta t$}\Comment{Move to next droplet bounce.}
    \State{$\eta \gets \text{IFFT}(\hat{\eta})$; plot $\eta$ and $\M{x}_p$.}
    \EndWhile
   
\end{algorithmic}
\end{algorithm}

Accordingly, we evolve the details of the wave profile using equations (\ref{ps_eq_1}–\ref{ps_eq_3}). We must, however, evolve the droplet trajectory of equation (\ref{ps_eq_3}) in physical space, so a single Fourier inversion in each Faraday period is required (see Algorithm \ref{alg1} for further details).

In both the central difference and pseudo-spectral methods, we utilize either a leapfrog or fourth-order Runge-Kutta (RK4) method to evolve the potential $\phi$ and the wave-field $\eta$ so as to ensure numerical accuracy and convergence \cite{shampine2009stability, butcher1996history}. We also ensure that both methods exhibit stability by imposing the Courant–Friedrichs–Lewy condition with respect to temporal and spatial discretization \cite{courant1928partiellen}.
\subsection{Numerical Results}

\begin{figure}[!t]
     \centering
    \begin{subfigure}[b]{0.49\textwidth}
        \centering
        \includegraphics[height=1.2in]{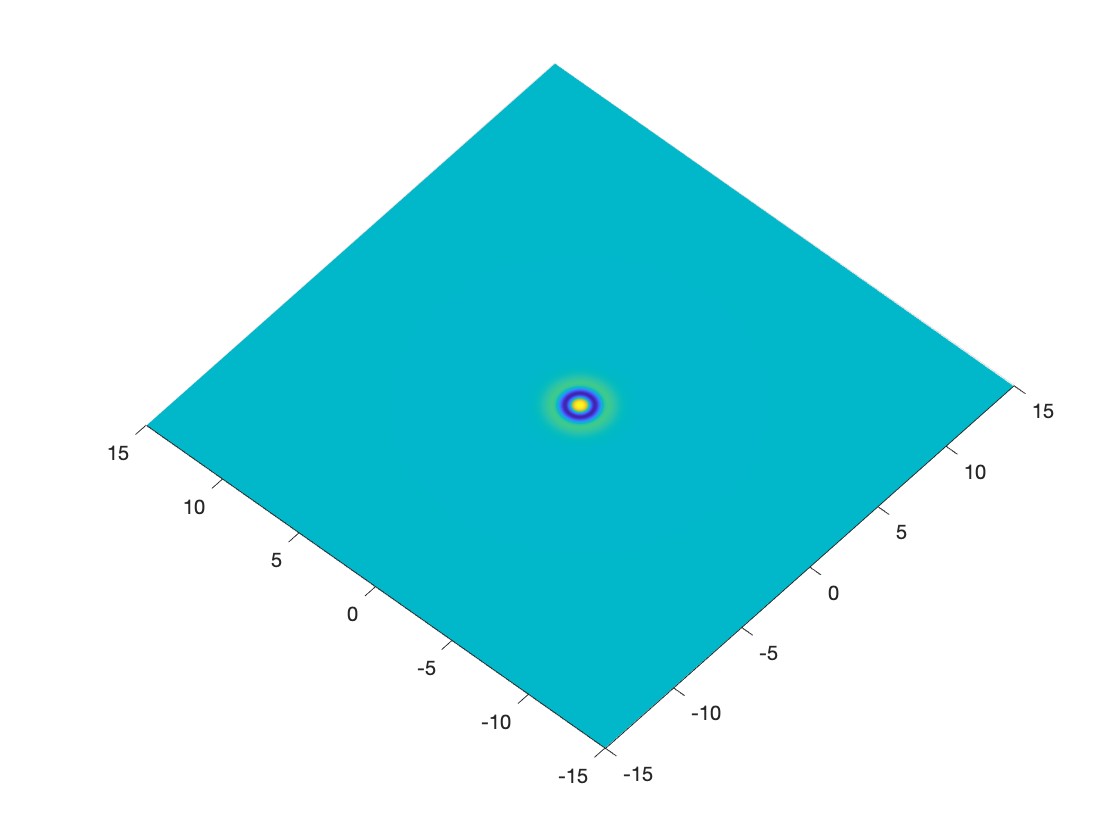}
        \caption{Initial static wave profile, shown prior to droplet impact.}
    \end{subfigure}%
    ~ 
    \begin{subfigure}[b]{0.49\textwidth}
        \centering
        \includegraphics[height=1.2in]{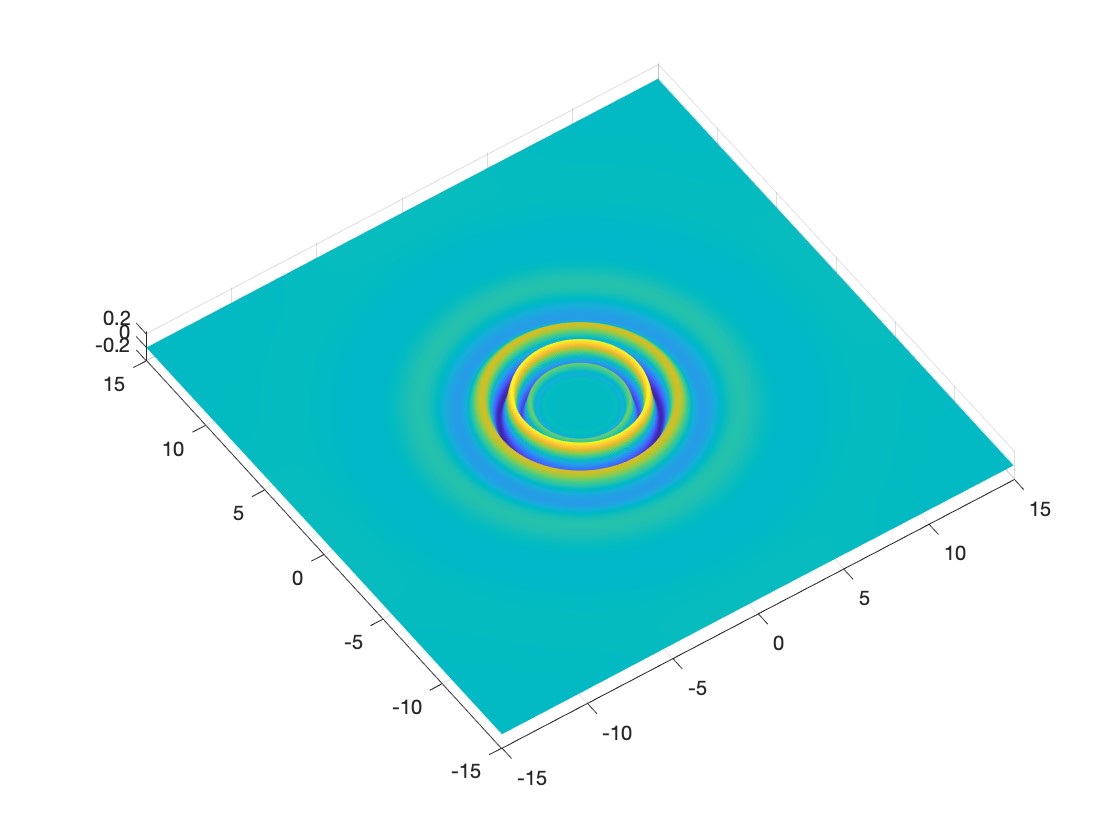}
        \caption{Wave profile, three Faraday periods after droplet impact.}
    \end{subfigure}
    
    \begin{subfigure}[b]{0.49\textwidth}
        \centering
        \includegraphics[height=1.5in]{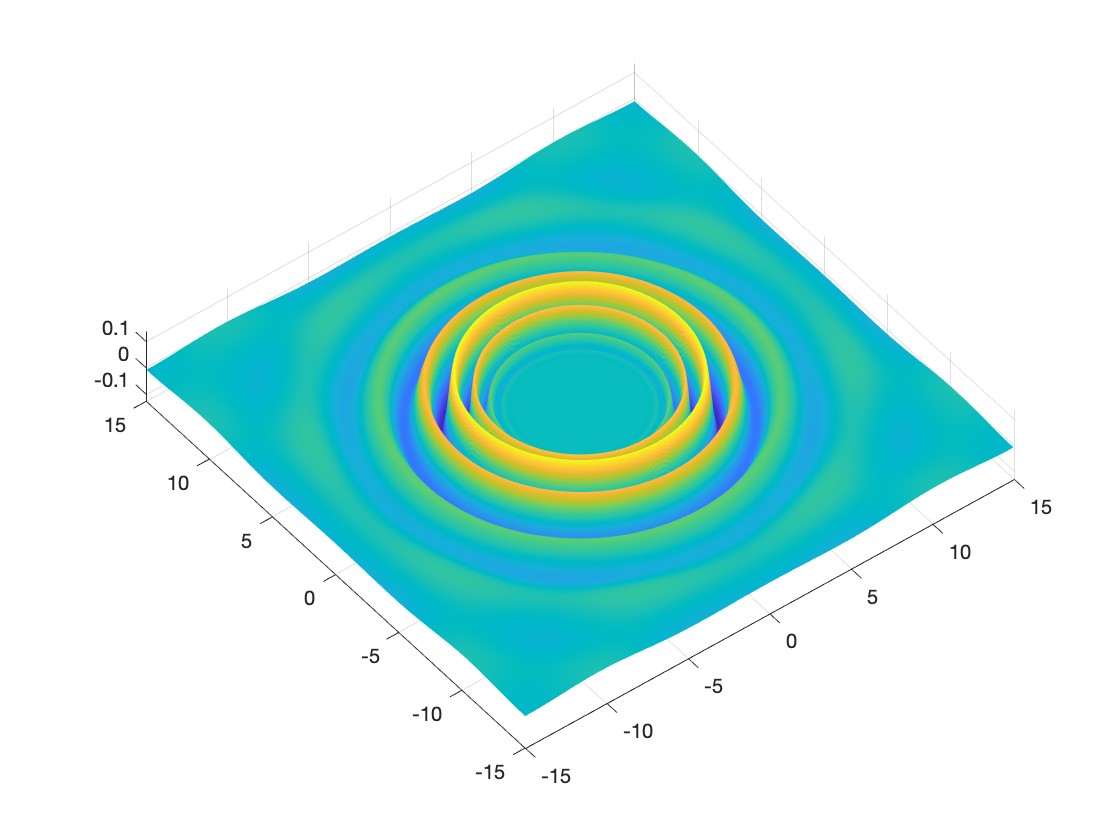}
        \caption{Wave profile, eight Faraday periods after droplet impact.  \label{fig:one_impact_c}}
    \end{subfigure}%
    ~
    \begin{subfigure}[b]{0.49\textwidth}
        \centering
        \includegraphics[height=1.5in]{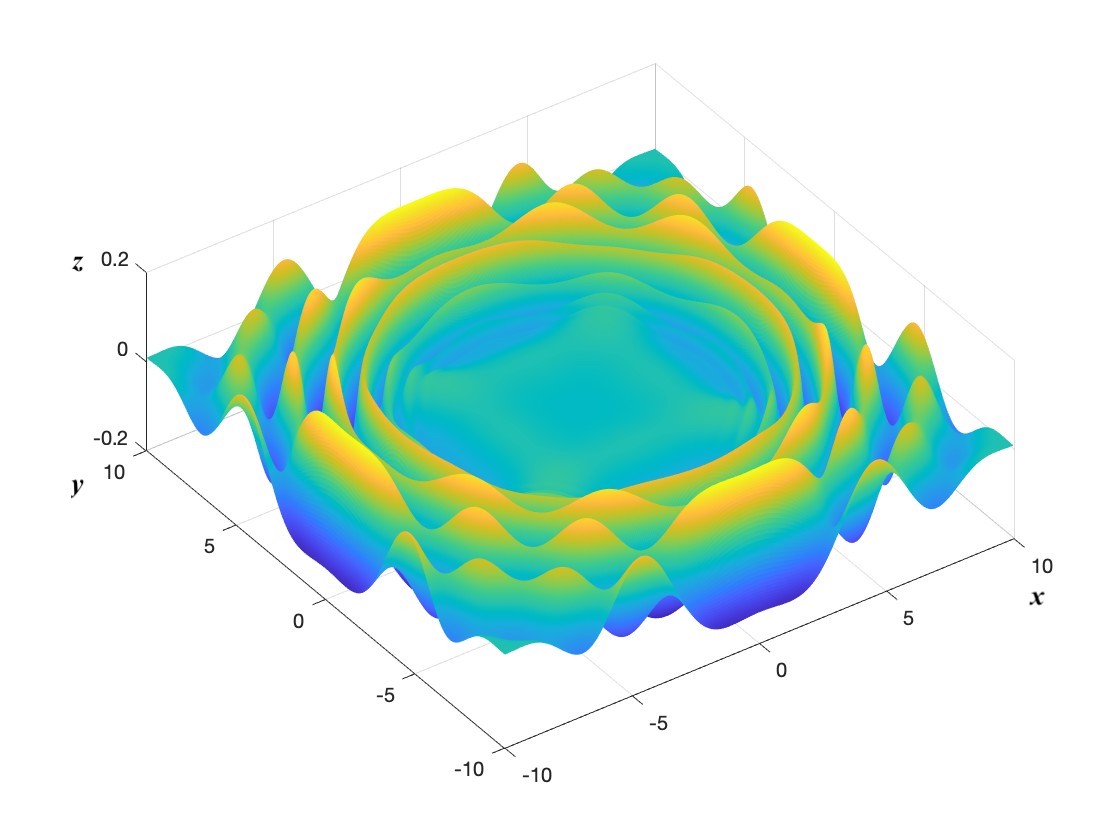}
        \caption{Wave profile, twenty Faraday periods after droplet impact.  \label{fig:one_impact_d}}
    \end{subfigure}
    \caption{A numerical model of the evolution of a droplet wave profile after a single droplet bounce at $(0,0,0)$, seen from a bird's eye view. The impact of a single droplet leads to the formation of surface waves. The wave propagates outward from its initial impact, displaying sharp parallels in overall form and scale to experimental results. The axis scaling is with respect to the Faraday wavelength, $\lambda_F$. Created by student researcher.}
    \label{fig:one_impact}
\end{figure}

We developed a MATLAB-based simulation library from scratch to implement the model described in Sections \ref{sec3} and \ref{sec4.1}. Our model is highly versatile and accurate in comparison to experiment and current, state-of-the-art models. We performed simulations of droplet dynamics for both the constant-bottom and variable-topography cases; fine-tuned the values of the forcing acceleration $\gamma$ to obtain both bouncing droplets (Figure \ref{fig:one_impact}, Figures \ref{fig7a}-\ref{fig7c}) and walking droplets (Figures \ref{fig7c}-\ref{fig7i}); and optimized the size of Fourier wave-vector bases.

In Figure \ref{fig:one_impact}, we demonstrate the impact of a droplet bouncing in a bath of constant depth. The radially symmetric surface waves of Figures \ref{fig:one_impact_c} and \ref{fig:one_impact_d} directly parallel the observations of experiment (such as in Figure \ref{fig-wave-field-for-intro}); additionally, the oscillatory decay rate of each wave in the radial direction occurs with an initial wavelength of $\lambda_F$, as evidenced in experimental conditions. 

\begin{figure}[!t]
    \centering
\begin{subfigure}[t]{0.33\textwidth}
\centering
  \includegraphics[width=\textwidth]{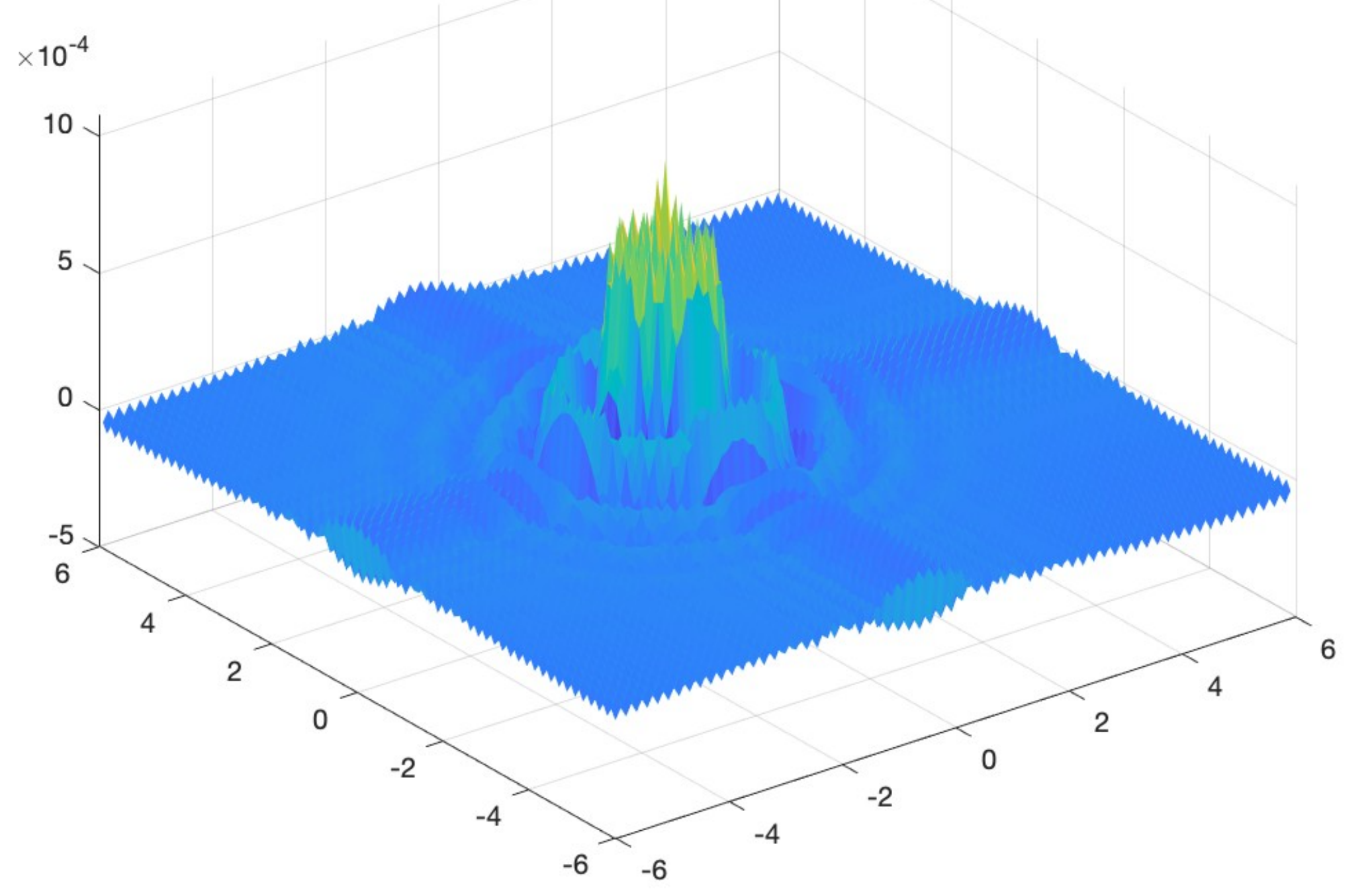}
\caption{Previous pilot-wave model. \label{fig7a}}
\end{subfigure}\hfill
\begin{subfigure}[t]{0.33\textwidth}
\centering
    \includegraphics[width=\textwidth]{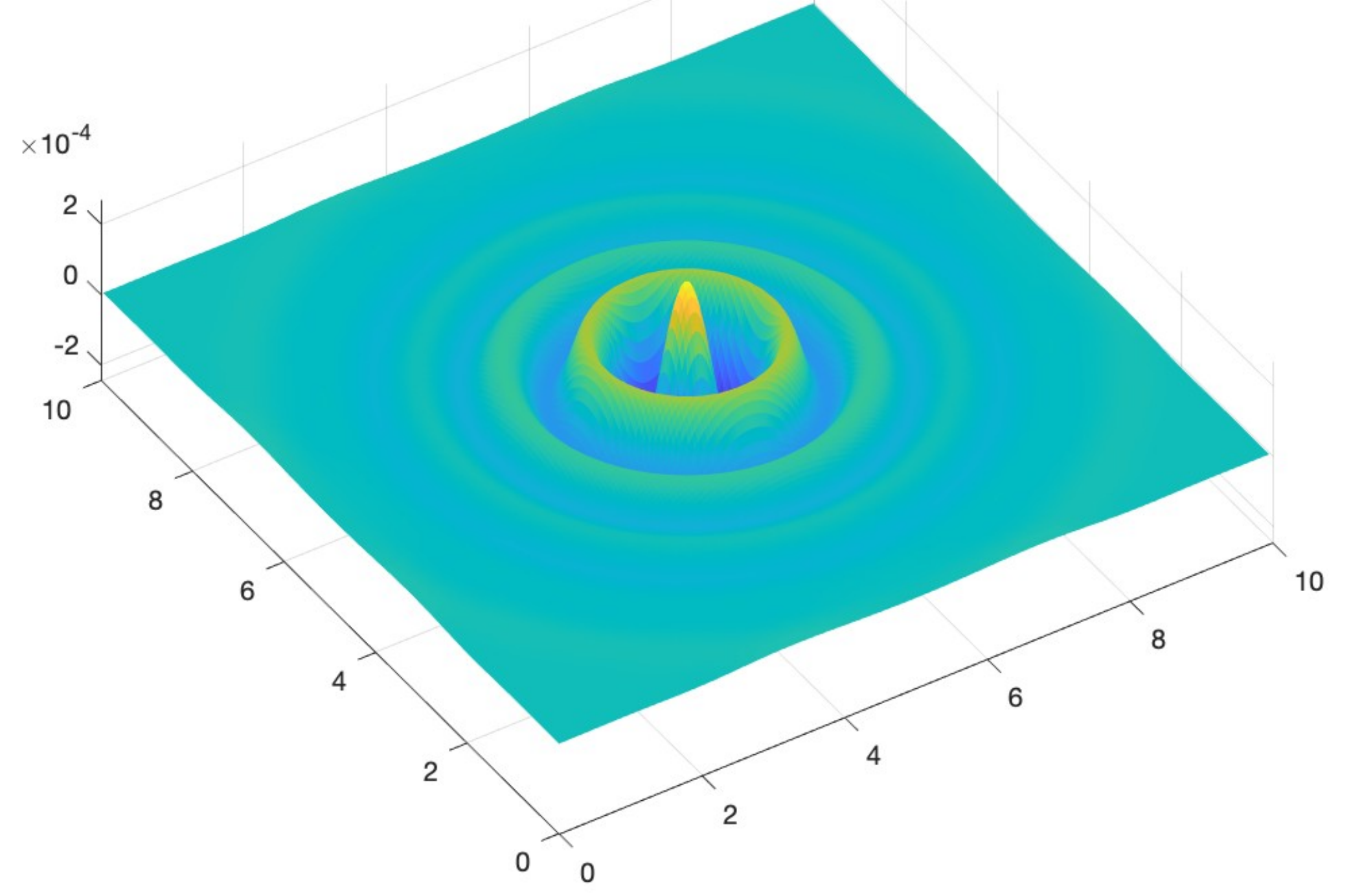}
\caption{Pseudo-spectral method.}
\end{subfigure}\hfill
\begin{subfigure}[t]{0.33\textwidth}
\centering
    \includegraphics[width=\textwidth]{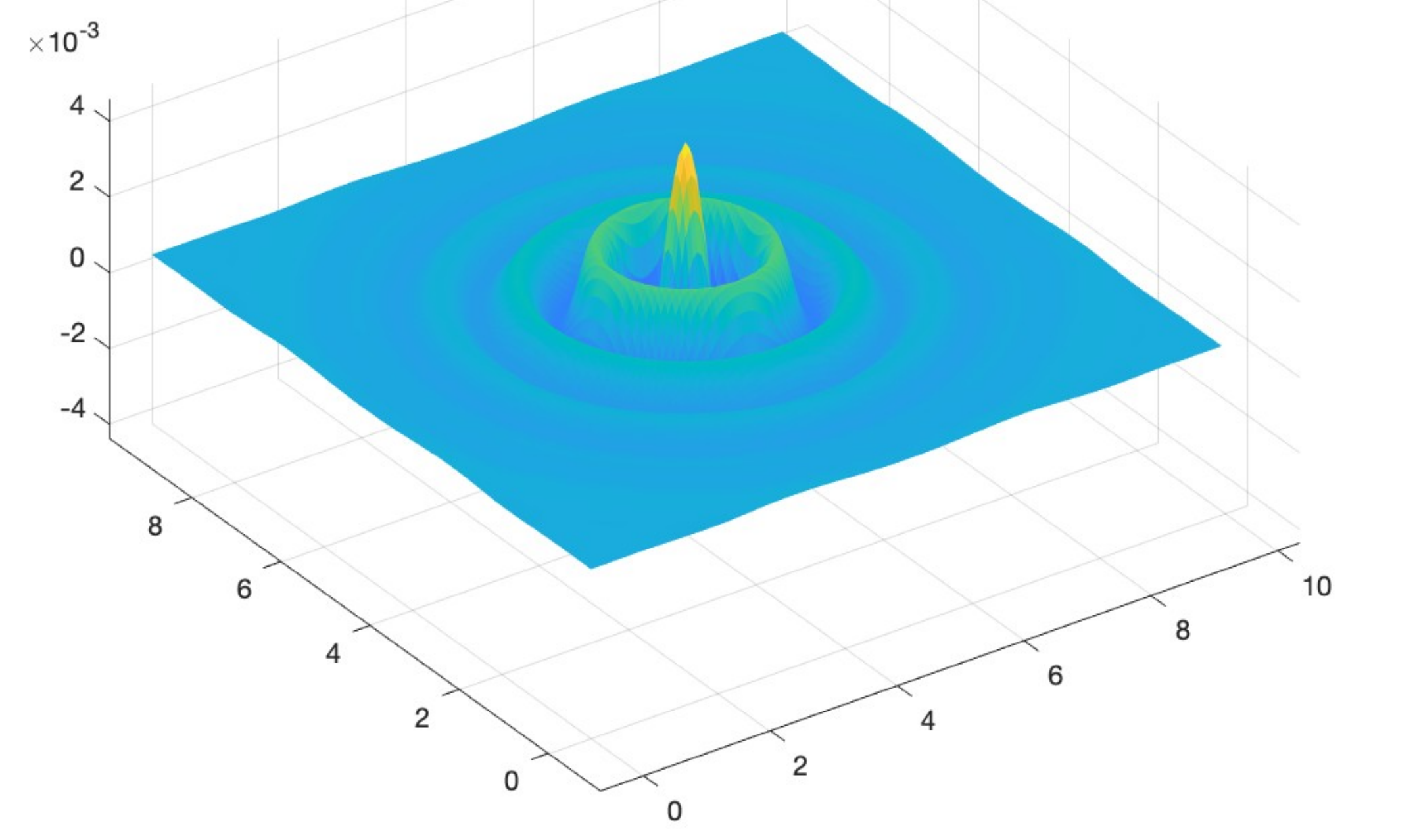}
\caption{Second-order central difference. \label{fig7c}}
\end{subfigure}\hfill

\begin{subfigure}[t]{0.33\textwidth}
\centering
  \includegraphics[width=\textwidth]{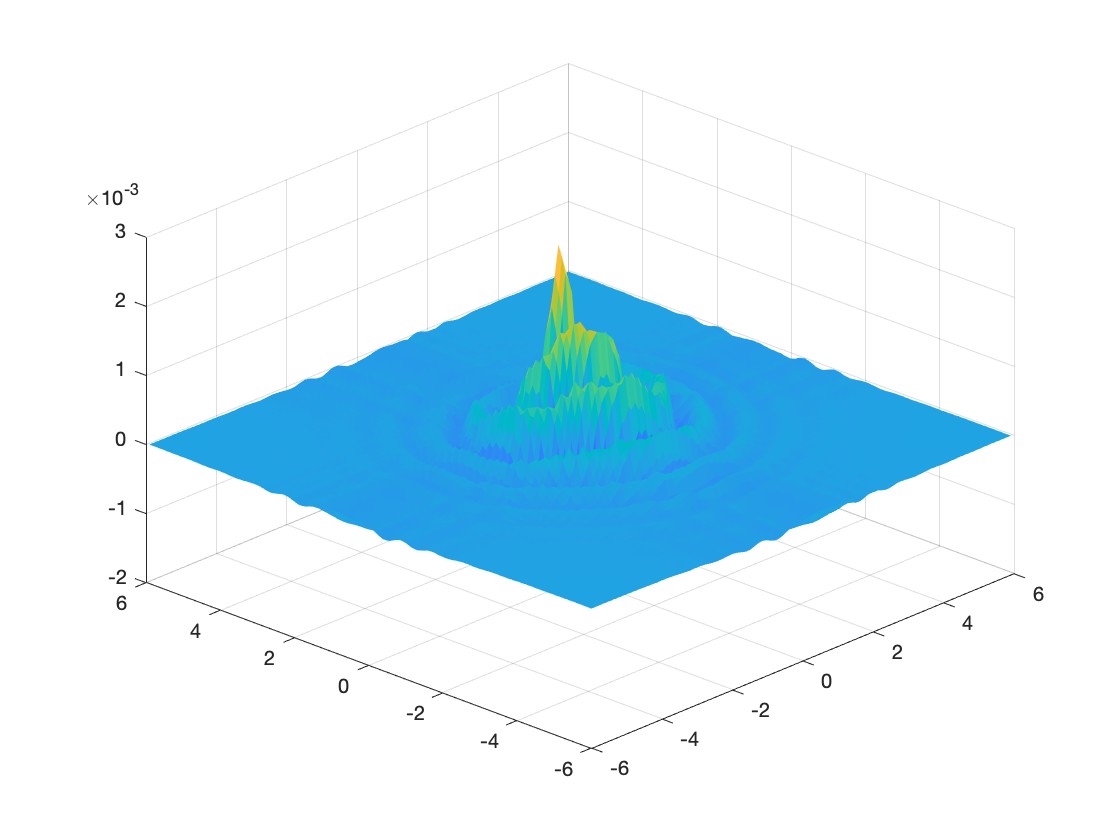}
\caption{Previous pilot-wave model. \label{fig7d}}
\end{subfigure}\hfill
\begin{subfigure}[t]{0.33\textwidth}
\centering
    \includegraphics[width=\textwidth]{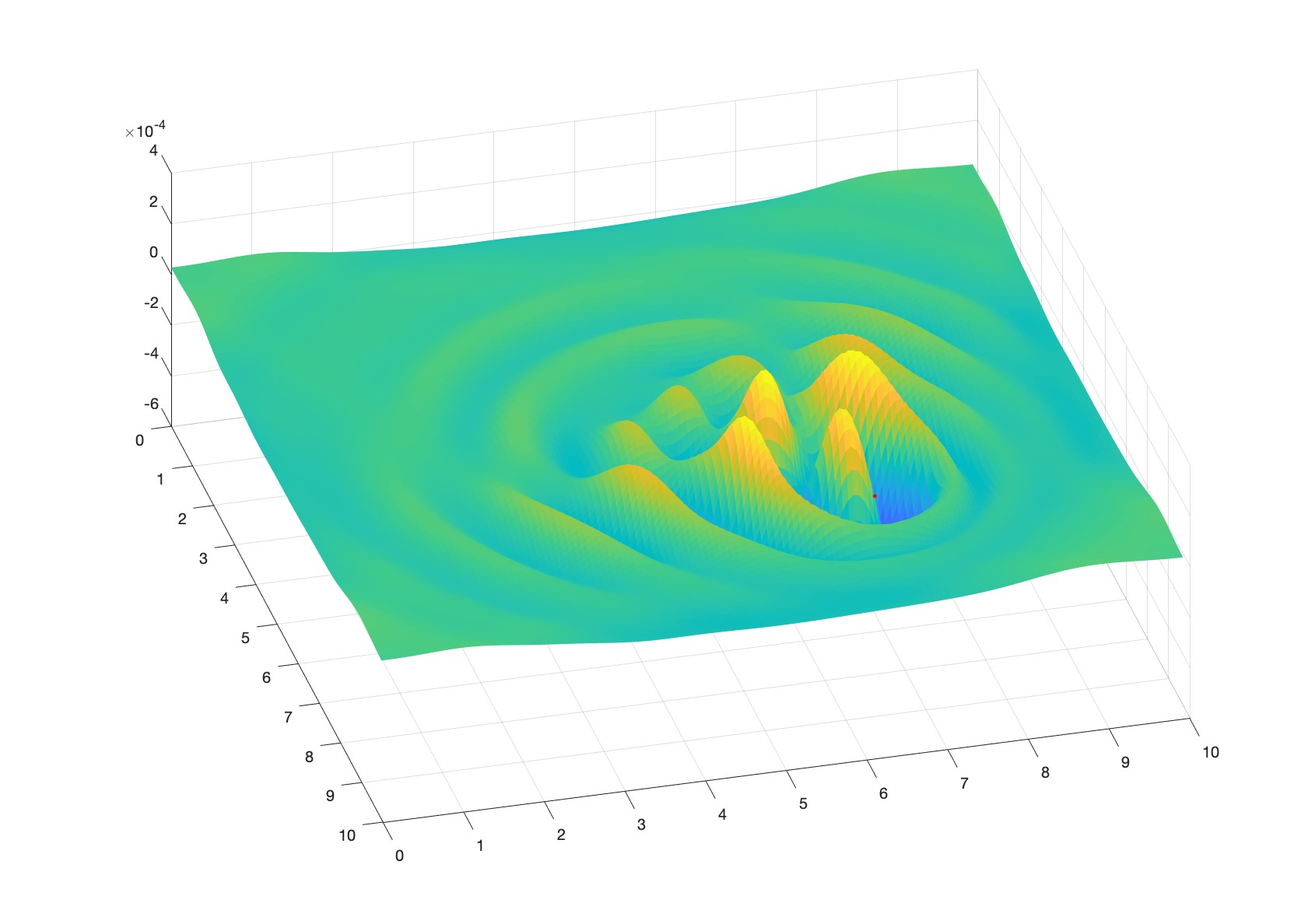}
\caption{Pseudo-spectral method.}
\end{subfigure}\hfill
\begin{subfigure}[t]{0.33\textwidth}
\centering
    \includegraphics[width=\textwidth]{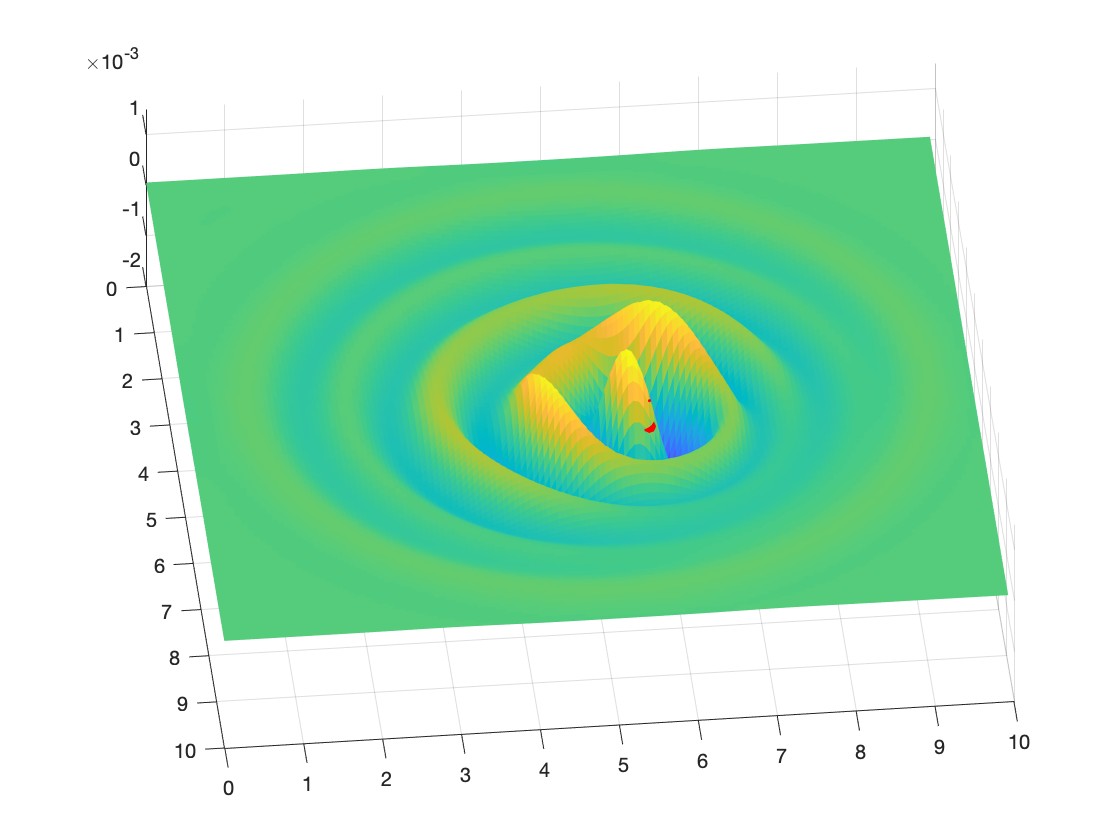}
\caption{Second-order central difference. \label{fig7f}}
\end{subfigure}\hfill

\begin{subfigure}[t]{0.33\textwidth}
\centering
  \includegraphics[width=\textwidth]{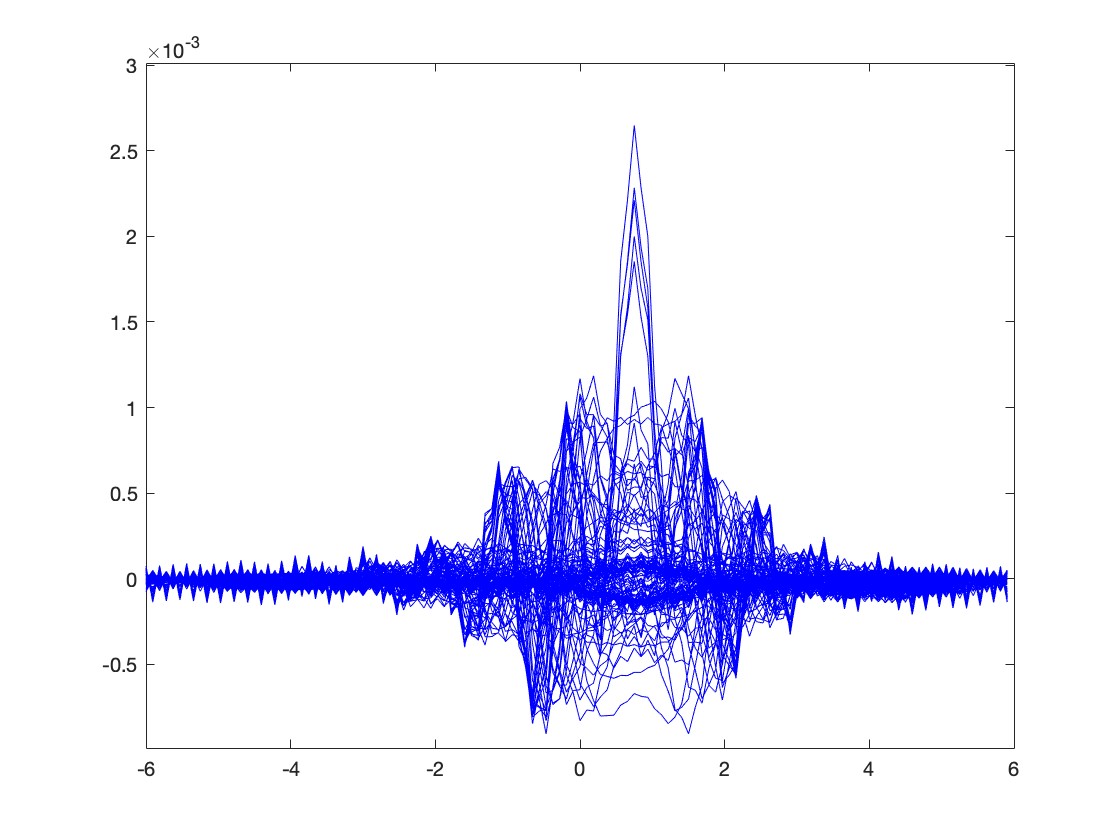}
\caption{Previous pilot-wave model. \label{fig7g}}
\end{subfigure}\hfill
\begin{subfigure}[t]{0.33\textwidth}
\centering
    \includegraphics[width=\textwidth]{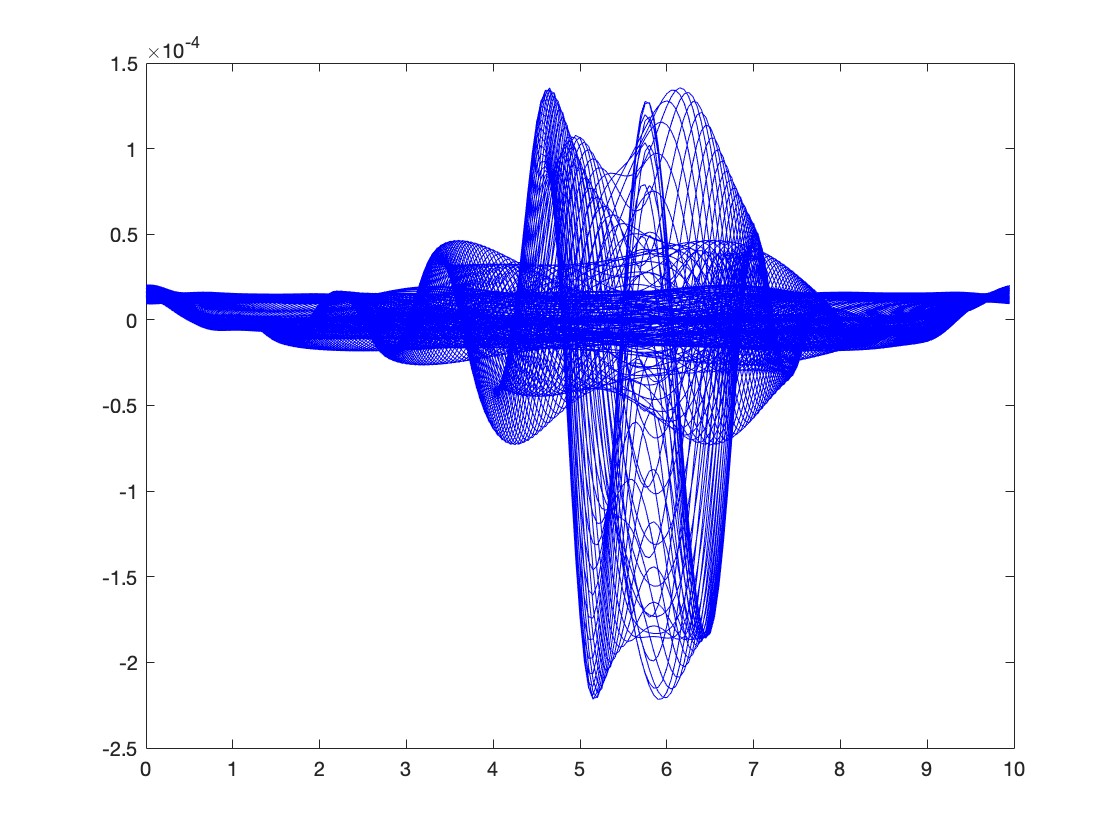}
\caption{Pseudo-spectral method.}
\end{subfigure}\hfill
\begin{subfigure}[t]{0.33\textwidth}
\centering
    \includegraphics[width=\textwidth]{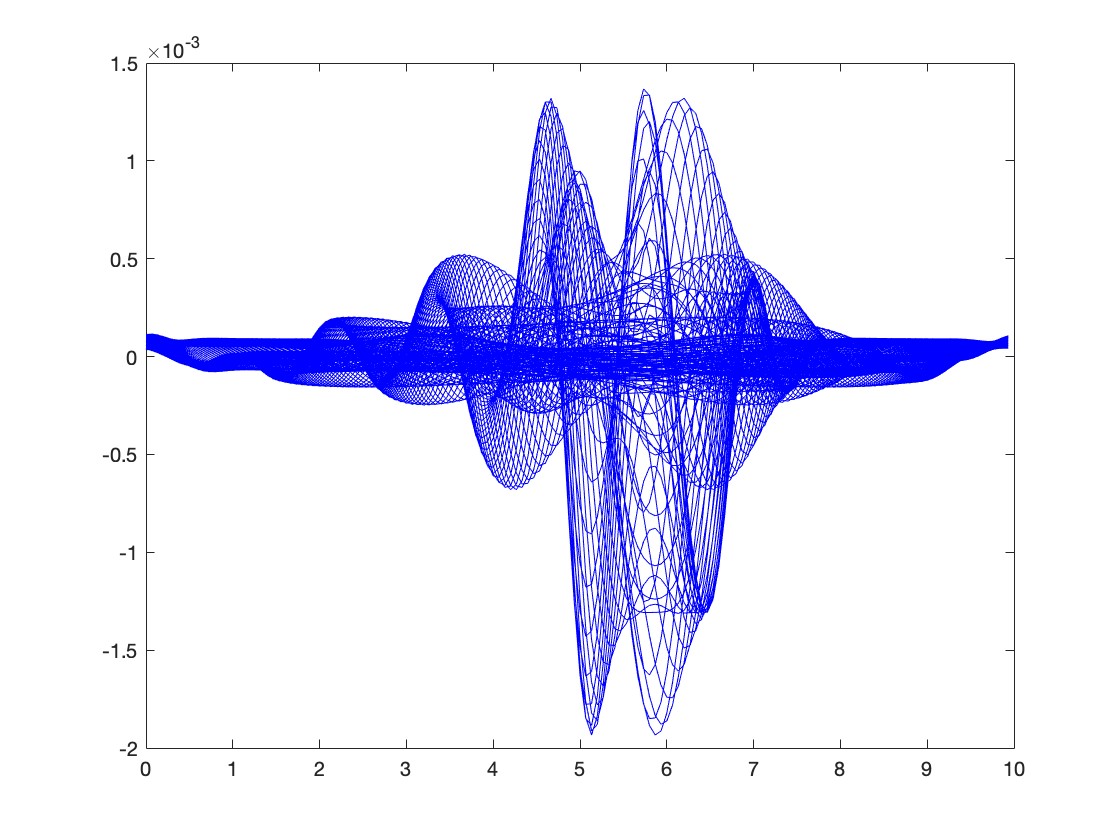}
\caption{Second-order central difference. \label{fig7i}}
\end{subfigure}\hfill
    \caption{A comparison of the ability of three models for 3D pilot-wave dynamics. We compared our results to the model developed by Faria \cite{faria2017model}, which is the existing state-of-the-art for droplet dynamics. Figures \ref{fig7a}-\ref{fig7c} display the wave-field resulting from five bounces of a \textit{bouncing} droplet; Figures \ref{fig7d}-\ref{fig7f} display the wave-field resulting from a \textit{walking} droplet; and Figures \ref{fig7g}-\ref{fig7i} display cross-sections of the wave-fields in Figures \ref{fig7d}-\ref{fig7f}. Created by student researcher.}
    \label{fig-big-big-big}
\end{figure}

In Figure \ref{fig-big-big-big}, we compare the results of our two numerical approaches with the state-of-the-art hydrodynamic model developed by Faria \cite{faria2017model}. Our methods of estimating the propagation of a walking droplet and the evolution of its wave-field in three dimensions exhibit significantly higher accuracy, smoothness, and qualitative similarity to experiment than state-of-the-art hydrodynamic models. Additionally, the decay rate of droplet velocities in the pseudo-spectral method most closely resembles that of experimental decay rates.

Although our current simulation capabilities exclusively utilize straightforward cavity geometries, our model allows for the most accurate known numerical characterization of droplet dynamics. Specifically, our model extends droplet dynamics to complex variable topographies in greater detail than existing studies, which either invoke lower-dimensional procedures \cite{nachbin2017tunneling}, omit droplet dynamics \cite{andrade2018three}, or invoke an imprecise approximation of the DtN operator \cite{faria2017model}.

\section{Experimental Results}
\label{sec5}

\begin{figure}[!t]
\begin{subfigure}[t]{0.49\textwidth}
\centering
 \includegraphics[width=0.5\linewidth]{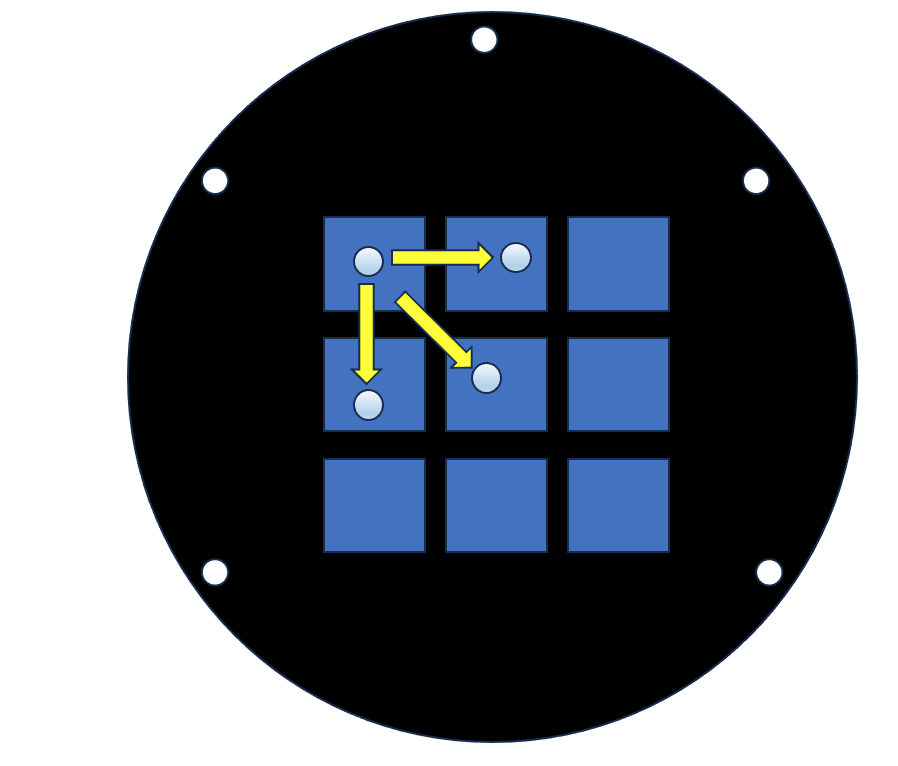}
\caption{Bird's-eye view.}
\end{subfigure}%
 ~
\begin{subfigure}[t]{0.49\textwidth}
\centering
  \includegraphics[height=1.2in]{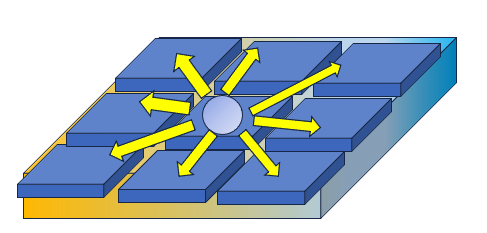}
\caption{Side profile.}
\end{subfigure}
    \centering
    \caption{A simplified diagram of the experimental setup. We present two views of the laser-cut cavity geometry. The droplet occasionally exhibits diagonal tunneling between cavities, although such instances of tunneling are rare. Created by student researcher.}
    \label{fig7}
\end{figure}

In addition to developing the theory of droplet dynamics within a three-dimensional region with varying bottom topography, we conducted experimental investigations of droplet tunneling between cavities. We now describe several experimental behaviors of interest.

\paragraph{Experimental Setup.} The experimental apparatus consisted of acrylic plates, a piezoelectric droplet generator, and a speaker providing vertical vibrations of the form $\gamma \cos(2 \pi \omega t)$. We utilized values of $\gamma$ close to the Faraday instability, and fixed $\omega = 80$ Hz.

To ensure the existence of bouncing and walking states (such as in Figure \ref{fig-wave-field-for-intro}), we utilized low-viscosity silicone oil, with dynamic viscosity $\mu \approx 2 \cdot 10^{-2}$ Pa$\cdot$s, surface tension $\sigma \approx 0.0209$ N$\cdot$m$^{-1}$, and density $\rho = 0.965 \, \text{kg}/{\text{m}^3}$. The depth of the silicone oil layer in each cavity was set to precisely 6 mm, with a 0.5 $\pm$ 0.05 mm thick oil film above each cavity to assist crossing. We constructed our laser-cut cavity geometry to optimize for tunneling probability, ensuring that wells were sufficiently wide to provide the droplet sufficient momentum to cross the barrier, yet not excessively large to the extent that the droplet became trapped in the same cavity or escaped the entire cavity system (see Figure \ref{fig7}). 

\begin{figure}[!t]
     \centering
    \begin{subfigure}[b]{0.31\textwidth}
        \centering
        \includegraphics[height=1.2in]{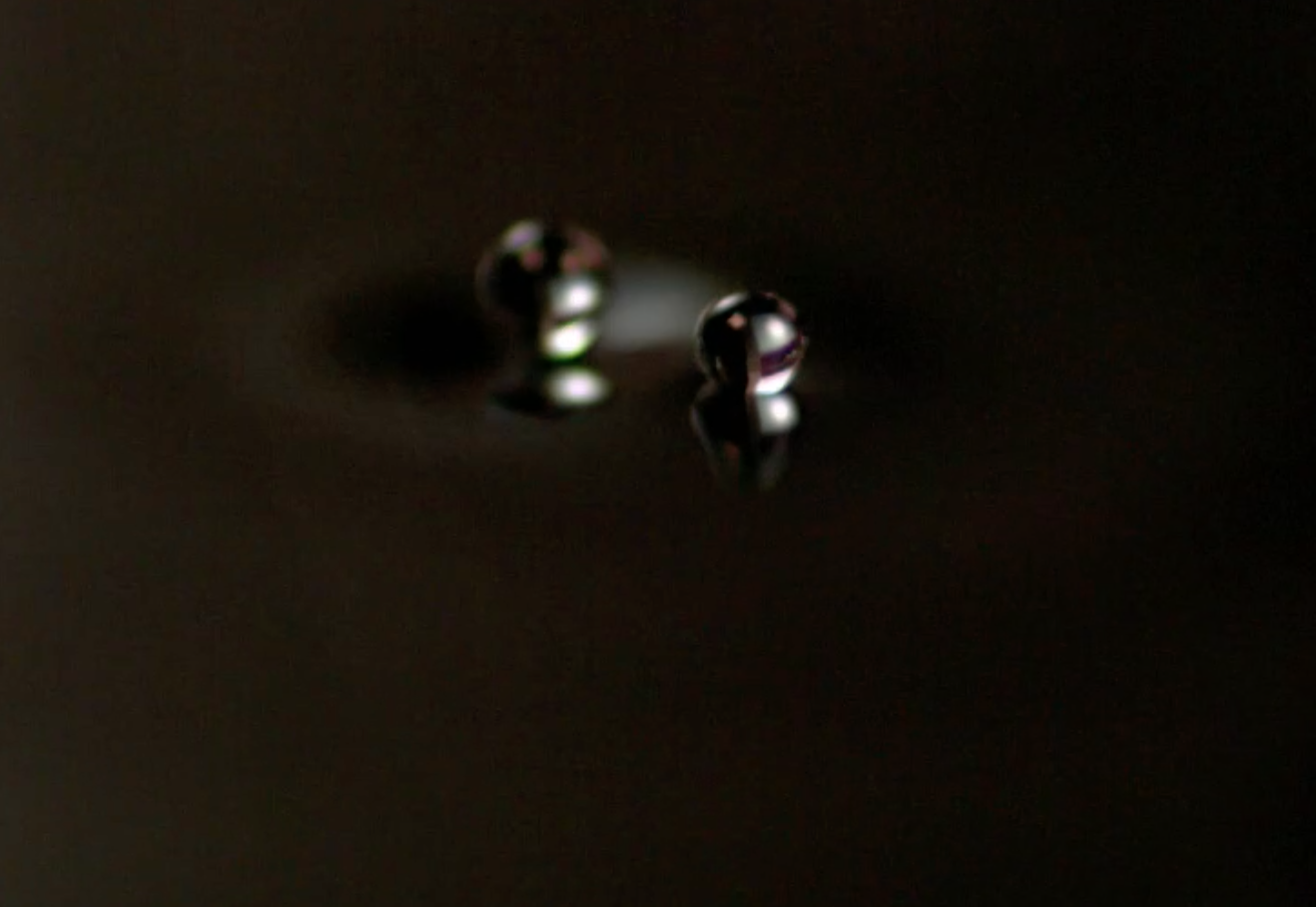}
        \caption{Droplet wave-field overlap exhibited prior to tunneling.}
    \end{subfigure}%
    ~ 
    \begin{subfigure}[b]{0.31\textwidth}
        \centering
        \includegraphics[height=1.2in]{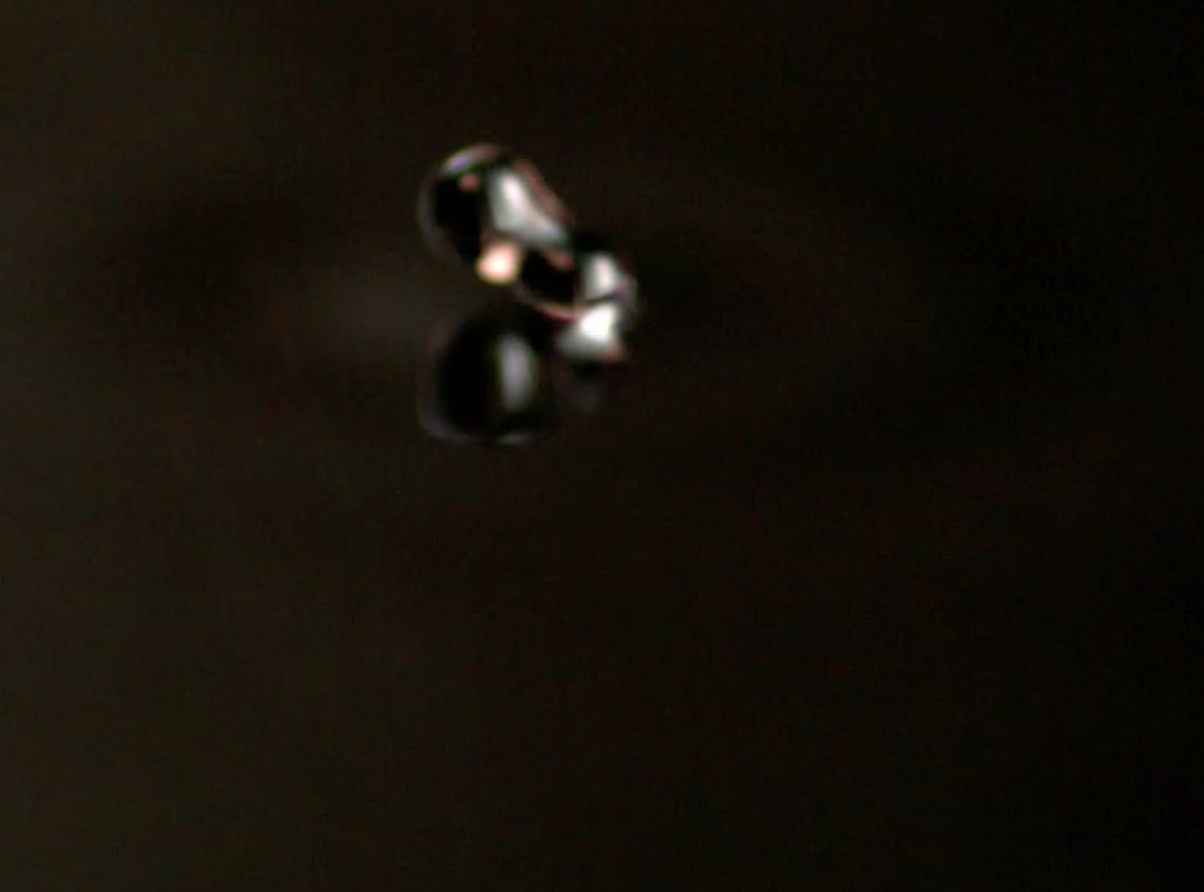}
        \caption{Two-droplet coalescence initiated after tunneling over barrier.}
    \end{subfigure}
    ~
    \begin{subfigure}[b]{0.31\textwidth}
        \centering
        \includegraphics[height=1.2in]{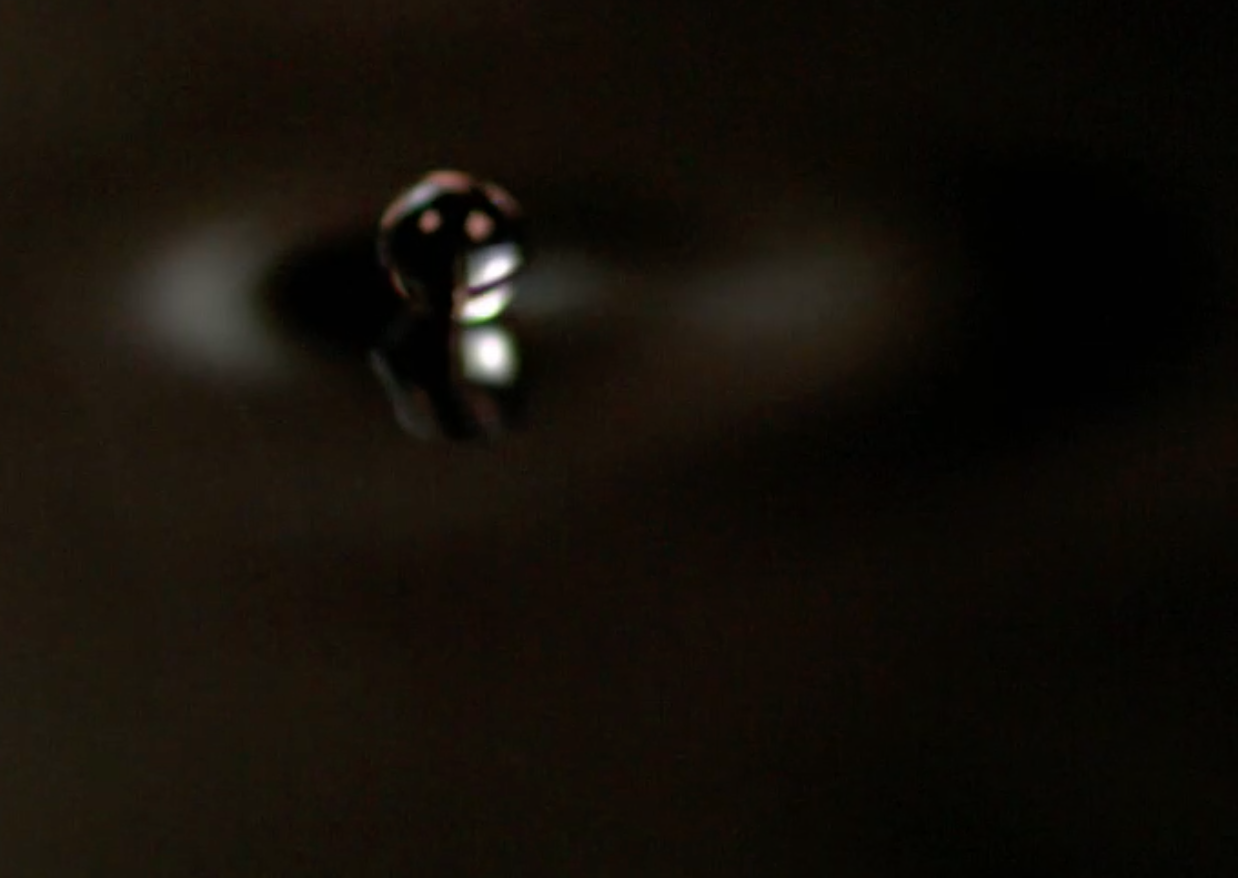}
        \caption{Single-droplet system results after merging process is complete.}
    \end{subfigure}
    \caption{Two droplets, with coupled wave-fields, exhibiting successful cooperative tunneling. The larger droplet attempts to force the smaller droplet over the barrier, then effectively merges with the smaller droplet to form a stationary bouncing droplet. Created by student researcher.}
    \label{fig:two_droplets}
\end{figure}

\paragraph{Experimental Results.} Tunneling possibilities were highly sensitive to values of $\gamma$ and cavity width; we determined the optimal cavity width for high crossing probabilities to be approximately 8.9 $\pm$ 0.05 mm at $\gamma = 4.19 \pm 0.02$ m/s$^2$. In such parameter regimes, we observed two droplets tunneling collectively across a barrier, with higher probability than that of individual droplet tunneling (see Figure \ref{fig:two_droplets}). Such tunneling effects may potentially have connections to the quantum-mechanical phenomenon of superradiance.

An investigation of the average rate of tunneling over three one-hour experiments demonstrated that the droplet tunneled at a rate of 9.3 barrier crossings per minute. Figure \ref{fig:heatmap} provides a description of the distribution of droplet positions through a relative comparison of durations in which the droplet occupied each well. The droplet consistently tunneled around exterior cavities without entering in the center cavity, suggesting the existence of a novel hydrodynamic quantum analogy to angular momentum.

\begin{figure}[!t]
    \centering
    \includegraphics[width=0.5\linewidth]{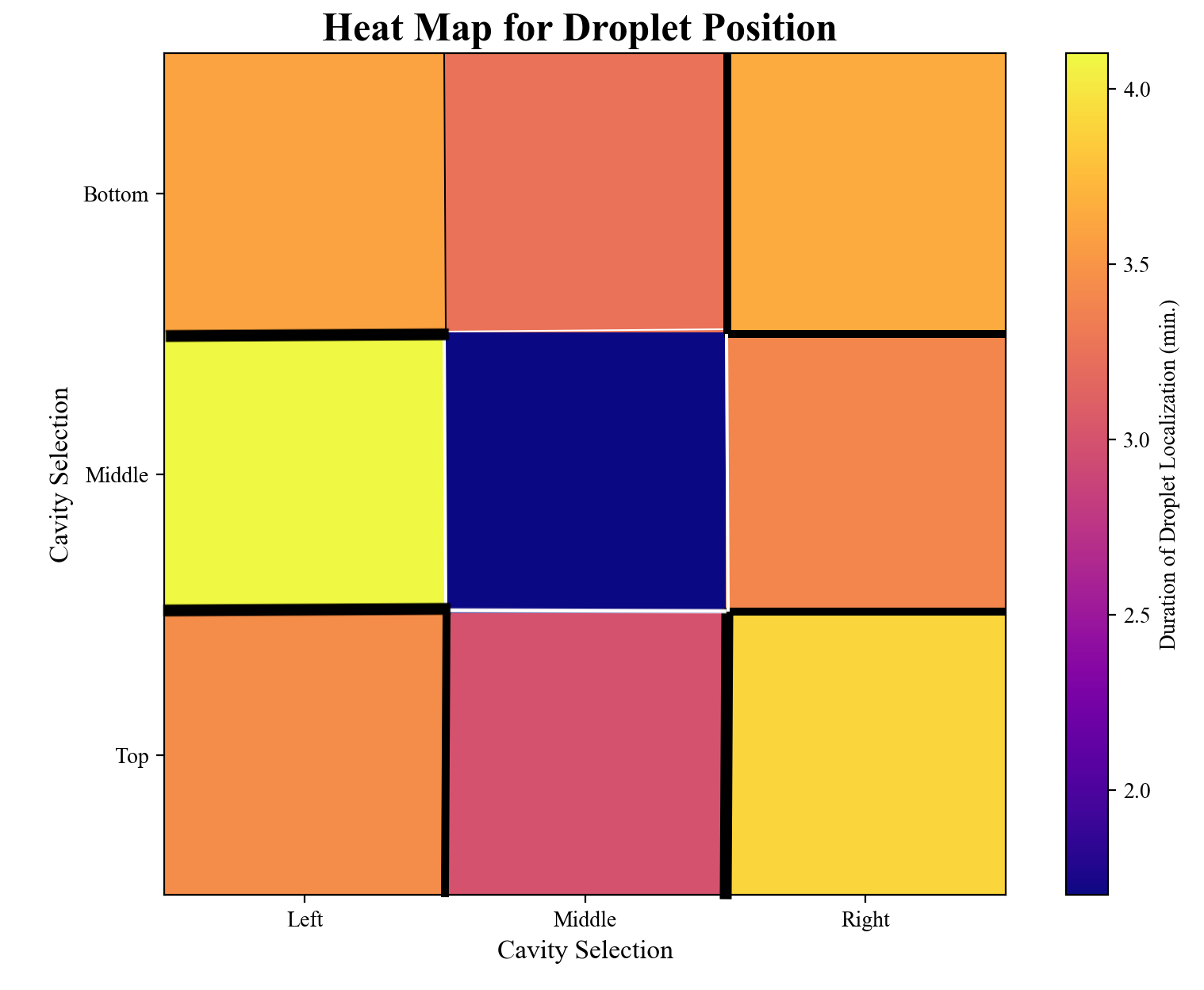}
    \caption{A distribution of the positions of the droplet over thirty minutes of tunneling, averaged over 10 trials. Each square cavity is given a “temperature" designation using the droplet's affinity for that particular cavity; such affinity is measured through the duration in which the droplet occupied the cavity. Additionally, the relative widths of the edges between cavities scale proportionally to the frequency of crossing with respect to each edge. Created by student researcher.}
    \label{fig:heatmap}
\end{figure}

\section{Conclusion}
\label{sec7}

\paragraph{Discussion and Key Takeaways.}

Our investigations provide a synthesis of theoretical, experimental, and numerical to understand droplet dynamics. In addition to experimental analysis of droplet tunneling in cavity geometries, we utilized our theoretical framework to characterize droplet hydrodynamics using a highly accurate model.

Our experimental studies provide a detailed characterization of novel tunneling-related phenomena. We observed that droplets were able to tunnel with higher probability when doing so cooperatively, and we discovered that droplets develop an effective angular momentum while tunneling – a previously unknown hydrodynamic quantum analogy.

From a theoretical standpoint, we resolved multi-scale dynamics involving coupled droplet and wave-field evolutions. Our numerical model exhibits greater accuracy and generalizability than existing models, allowing for a vast applicability in dynamical analysis of droplet behaviors \cite{bush2015pilot} Additionally, our numerical model for non-local wave interactions over variable topographies in shallow-water regimes, especially the DtN formalism, may be scaled into a broader framework to describe coastal ocean wave dynamics.

\paragraph{Future Work.} We hope to extend our simulations to include more complex cavity geometries in order to numerically describe tunneling probabilities. Additionally, we would like to be able to use compute clusters to accelerate the process of time-evolution for our simulations. Finally, we hope to investigate how effectively droplets bound in a so-called \textit{promenading} state \cite{arbelaiz2018promenading} exhibit cooperative tunneling at differing incidence angles.

\pagestyle{empty}
\bibliographystyle{style}

\bibliography{biblio}

%

\appendix
\gdef\thesection{Appendix \Alph{section}}

\end{spacing}
\end{document}